\documentclass[11pt]{article}

\usepackage{amsmath}
\usepackage{amssymb} 
\usepackage{caption}
\usepackage{graphicx}
\usepackage{authblk}
\usepackage[hypertexnames=false,colorlinks=true,linkcolor=blue,citecolor=blue]{hyperref}
\usepackage[numbers,comma,square,sort&compress]{natbib}
\usepackage[a4paper,text={6.5in,10in},centering]{geometry}

\graphicspath{{eps/}{pdf/}{images/}}

\makeatletter\@addtoreset{equation}{section}\makeatother

\newtheorem{Lemma}{Lemma}[section]

\newtheorem{Proposition}[Lemma]{Proposition}

\newenvironment{Proof}%
 {\begin{trivlist} \item[]{\bf Proof. }}%
 {\hspace*{\fill}$\rule{.4\baselineskip}{.4\baselineskip}$\end{trivlist}}

 {\begin{trivlist}\item[]\textbf{Acknowledgments.}}{\end{trivlist}}
\newtheorem{lemma}{Lemma}
\newtheorem{hypo}{Hypothesis}

\newtheorem{remark}{Remark}

\newenvironment{proof}[1][.]%
{\begin{trivlist}\item[]\textbf{Proof#1 }}%
 {\hspace*{\fill}$\rule{0.3\baselineskip}{0.35\baselineskip}$\end{trivlist}}

\def\Re{\mathop{\mathrm{Re}}}

\newcommand{\D}{\mathrm{d}}

\begin{document}

\title{Invasion Fronts Outside the Homoclinic Snaking Region in the Planar Swift-Hohenberg Equation}
\author[1]{David J.B. Lloyd}
\affil[1]{\small Department of Mathematics, University of Surrey, Guildford, GU2 7XH, UK}
\date{\today}
\maketitle

\begin{abstract}
It is well-known that stationary localised patterns involving a periodic stripe core can undergo a process that is known as `homoclinic snaking' where patterns are added to the stripe core as a bifurcation parameter is varied. The parameter region where homoclinic snaking takes place usually occupies a small region in the bistability region between the stripes and quiescent state. Outside the homoclinic snaking region, the localised patterns invade or retreat where stripes are either added or removed from the core forming depinning fronts. It remains an open problem to carry out a numerical bifurcation analysis of depinning fronts. 
In this paper, we carry out a numerical bifurcation analysis of depinning of fronts near the homoclinic snaking region, involving a spatial stripe cellular pattern embedded in a quiescent state, in the two-dimensional Swift-Hohenberg equation with either a quadratic-cubic or cubic-quintic nonlinearity. We focus on depinning fronts involving stripes that are orientated either parallel, oblique and perpendicular to the front interface, and almost planar depinning fronts. We show that invading parallel depinning fronts select both a far-field wavenumber and a propagation wavespeed whereas retreating parallel depinning fronts come in families where the wavespeed is a function of the far-field wavenumber. Employing a far-field core decomposition, we propose a boundary value problem for the invading depinning fronts which we numerically solve and use path-following routines to trace out bifurcation diagrams. We then carry out a thorough numerical investigation of the parallel, oblique, perpendicular stripe, and almost planar invasion fronts.
We find that almost planar invasion fronts in the cubic-quintic Swift-Hohenberg equation bifurcate off parallel invasion fronts and co-exist close to the homoclinic snaking region. Sufficiently far from the 1D homoclinic snaking region, no almost planar invasion fronts exist and we find that parallel invasion stripe fronts may regain transverse stability if they propagate above a critical speed. Finally, we show that  depinning fronts shed light on the time simulations of fully localised patches of  stripes on the plane. 
The numerical algorithms detailed have wider application to general modulated fronts and reaction-diffusion systems.

\end{abstract}

\section{Introduction}

\begin{figure}[htb]
	\centering
	\includegraphics[width=0.8\linewidth]{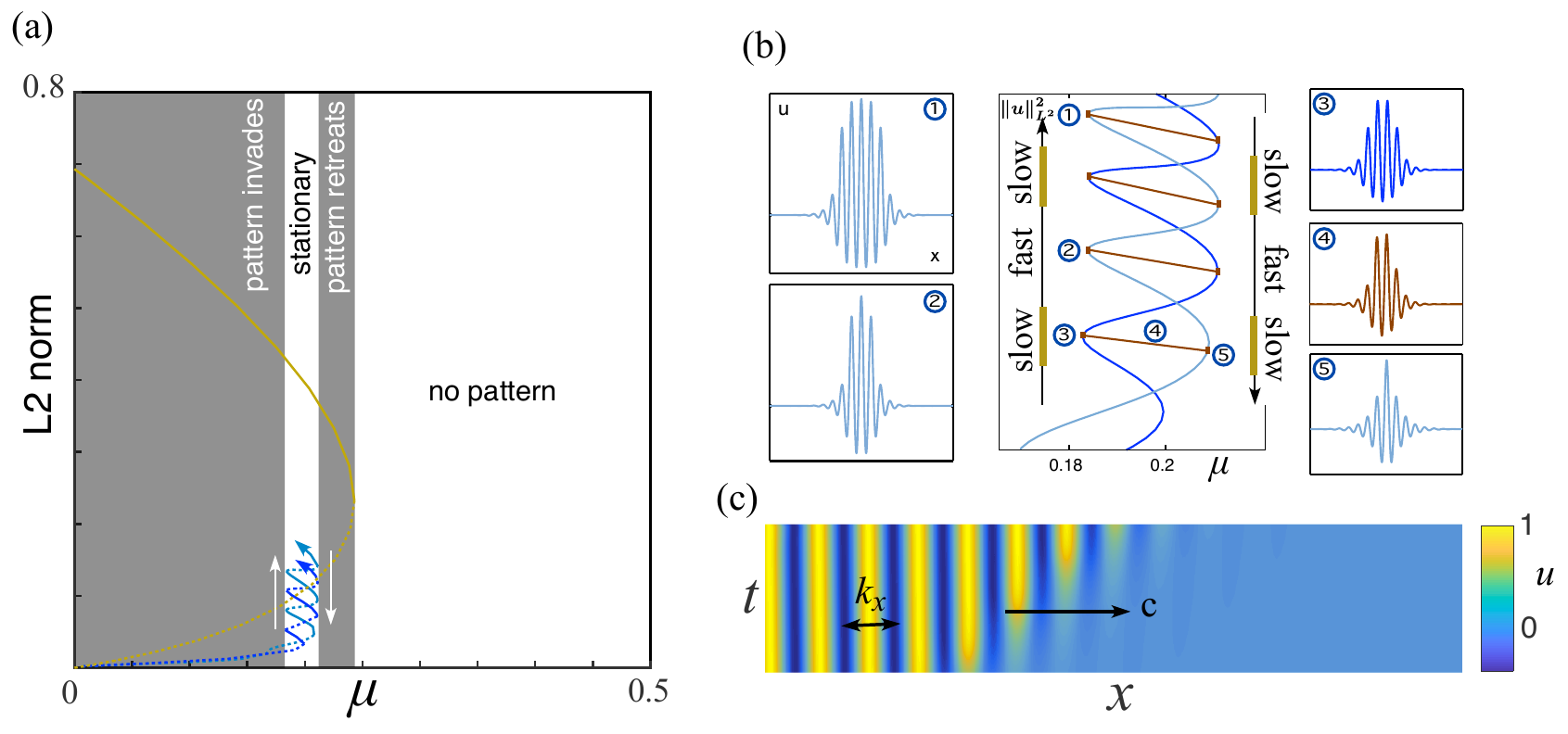}
	\caption{(a) Bifurcation diagram for 1D stationary localised and stripe patterns bifurcating from the trivial state for the quadratic-cubic SH equation with $\nu=1.6$. The homoclinic snaking region occurs between $0.184<\mu<0.211$ where stationary localised pulses are found shown in (b). In panel (b) we also plot how the depinning of the localised patterns occurs either side of the snaking region with a fast-slow propagation speed as the depinning passes each fold of the snake. 
A space-time plot of an invading front is shown in (c) and the pattern selection of the front propagation speed $c$ and far-field wavenumber $k_x$.
\label{f:SH_23_bif} }
\end{figure}
The formation of stationary localised patterns involving a spatially periodic core embedded in a a quiescent state has been found to occur in a range of applications from elastic buckling~\cite{thompson2015}, magnetic fluids~\cite{lloyd2015}, doubly diffusive convection~\cite{beaume2018}, nonlinear optics~\cite{Ackemann2009}, crystals~\cite{emmerich2012} and desertification~\cite{meron2018}, to name but a few. One of the most successful theories for their formation is known as homoclinic snaking~\cite{woods1999,coullet2000}. The prototypical system that possesses homoclinic snaking is the Swift-Hohenberg (SH) equation
\begin{equation}\label{e:sh}
u_t = -(1+\Delta)^2u - \mu u + f(u),
\end{equation} 
where $u=u(x,y,t)$, $\Delta$ is the planar Laplacian, $\mu,\nu\in\mathbb{R}$ are parameters and $f(u) = \nu u^2 - u^3=:f_{23}(u)$ (which we call the quadratic-cubic SH equation) or $f(u)=\nu u^3 - u^5=:f_{35}(u)$ (which we call the cubic-quintic SH equation). 
The classic bifurcation diagram for stationary states in 1D is shown in figure~\ref{f:SH_23_bif}(a) depicting the quiescent state undergoing a subcritical Turing instability at $\mu=0$ from which emerges an unstable branch of spatially periodic cellular patterns and two stationary localised states. The unstable spatially periodic patterns undergo a fold and restabilise to form a bistable region with the quiescent state. The unstable localised states however undergo an infinite number of folds forming two intertwining ``homoclinic snakes". At every other fold of the homoclinic snake, a new periodic pattern is added to the core of the localised state. A significant amount of theory has been developed to fully understand the homoclinic snaking mechanism~\cite{beck2009,makrides2014,dawes2010} as well as a large number of numerical investigations; see~\cite{knobloch2008,knobloch2015} for a review. The reason for the homoclinic snaking region is due to the existence of two stationary fronts (known as Pomeau fronts), connecting the trivial state and the stripe pattern and vice versa, that are {\it pinned} in an open region in parameter space~\cite{pomeau1986}. A localised pulse with a periodic core can be thought of as the gluing of these two Pomeau fronts~\cite{beck2009}. 

Outside of the homoclinic snaking region, the stationary localised Pomeau fronts start to move and either invade the quiescent state or retreat by adding or subtracting (respectively) a stripe at a time located at the interface between the stripes and the quiescent state~\cite{burke2006}. This process is known as {\it depinning} and occurs in the majority of the bistability region; see figure~\ref{f:SH_23_bif}(a). Localised pulses involving a periodic core, undergo exactly the same behaviour except the depinning occurs in opposite directions. The moving fronts are a special form of modulated fronts; see~\cite{sandstede2004}. While the stationary homoclinic snaking region is well understood numerically and analytically, less is known about the depinning region. The main problem is that the depinning fronts { do not} propagate at a constant speed (in fact they ``lurch") and so are not described by an ODE in the travelling frame. The lurching invasion process is due to passing near the folds of the homoclinic snake where the front slows down as the structure passes a fold and then speeds up; see figure~\ref{f:SH_23_bif}(b) \& (c). In~\cite{burke2006}, they carried out a formal, semi-analytical perturbation analysis near the edges of the 1D homoclinic snaking region.  However, one would like to numerically compute the depinning fronts and their pattern selection away from the snaking region and detect bifurcations. This is the main aim of the paper.  Furthermore, the weakly nonlinear analysis (see for instance Ponedel~{\it et al.}~\cite{ponedel2017}) suggests there is pattern selection mechanism for both invading and retreating fronts where they select a unique propagation speed $c$ and far-field wavenumber $k_x$ but the semi-analytical perturbation analysis does not say anything about a selection mechanism. Hence, we wish to also clarify when such a pattern selection mechanism occurs or not. 

\begin{figure}[h]
	\centering
	\includegraphics[width=0.8\linewidth]{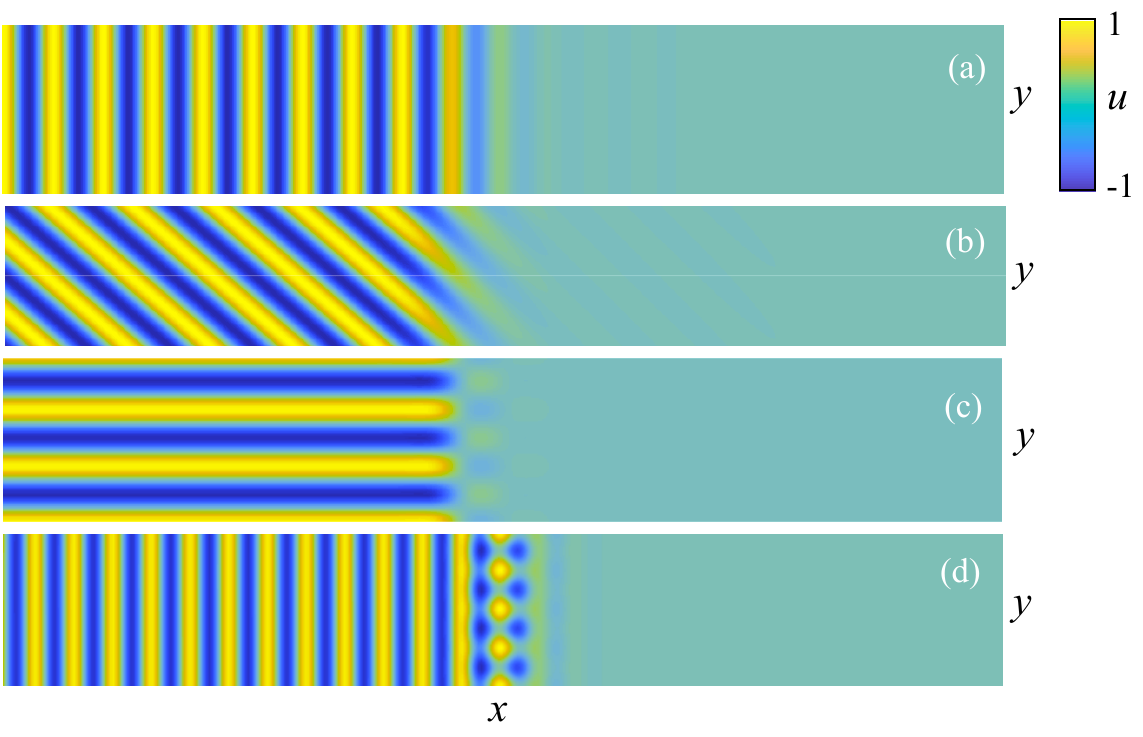}
	\caption{Four types of stripe fronts we focus on in this paper with stripes (a) parallel to the front interface (these fronts select a non-constant front speed and far-field wavenumber) (b) oblique to the front interface with constant front speed (these fronts select a constant propagation speed and far-field wavenumber selected but angle free) (c) perpendicular to the front interface  (these fronts select a constant front speed)  and (d) ``almost-planar" fronts (these fronts select a non-constant front speed and far-field wavenumber). \label{f:stripe_fronts}}
\end{figure}
In 2D, the range of possible depinning fronts increases dramatically. On the 2D-plane, there are three types of possible depinning stripe fronts: parallel, oblique, perpendicular and ``almost-planar" stripe fronts; see figure~\ref{f:stripe_fronts}. Parallel stripe fronts have stripes that are parallel to the front interface and similarly for the oblique and perpendicular stripe fronts. Almost-planar stripe fronts involve parallel stripes in their far-field but have a non-trivial spatially periodic structure at their interface in the transverse direction. Parallel stripe fronts are the simple 2D extension of the 1D invasion fronts shown in figure~\ref{f:SH_23_bif}(c) while the perpendicular, oblique, and almost-planar fronts are genuinely 2D. We note that stationary perpendicular stripe fronts were investigated in~\cite{avitabile2010} but the oblique stripe fronts have not been studied before. 

\begin{figure}
\centering
	\includegraphics[width=0.8\linewidth]{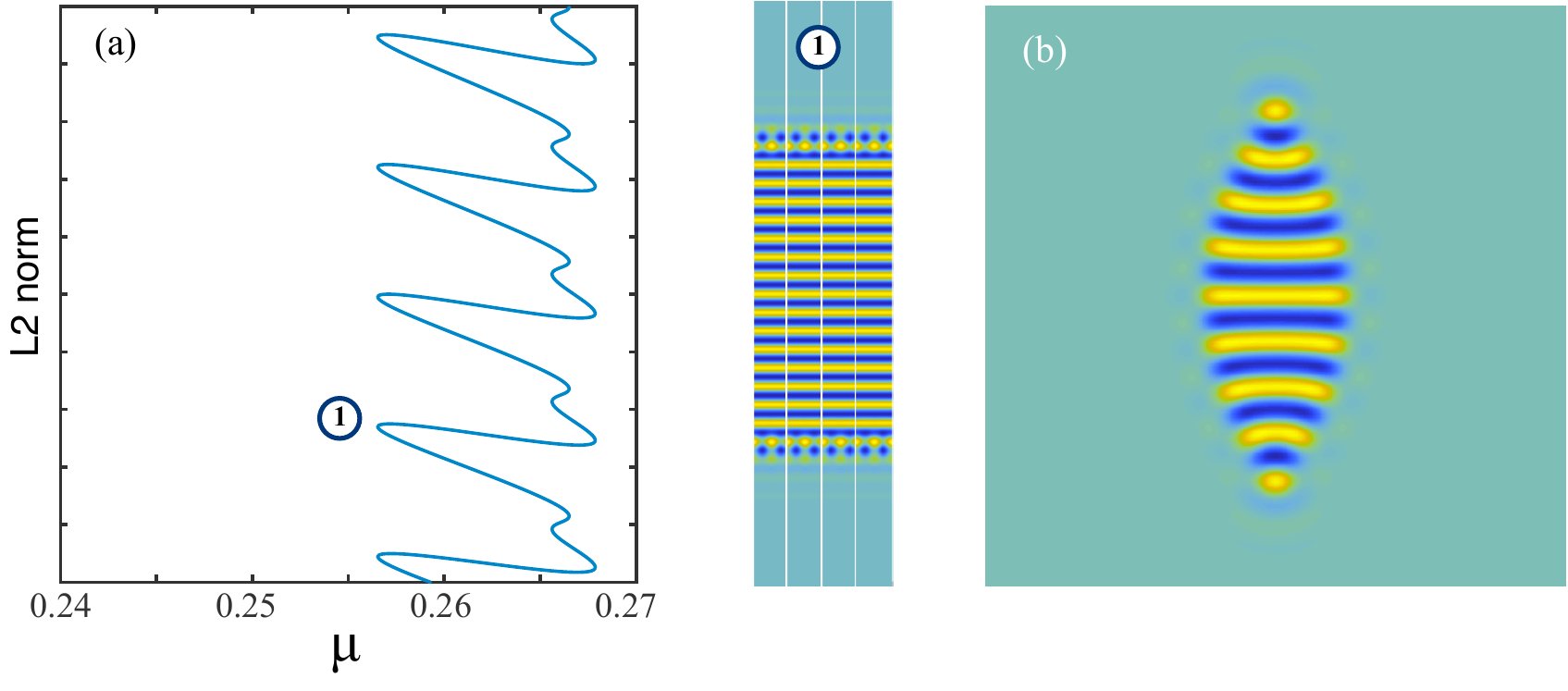}
	\caption{(a) Snaking of almost planar fronts shown inset \raisebox{.5pt}{\textcircled{\raisebox{-.9pt} {1}}} for the cubic-quintic SH equation with $\nu=1.25$. The snaking occurs between $0.256<\mu<0.268$ with two other folds at $\mu=0.2656$ and $\mu=0.2665$. (b) Stationary stable localised patches  worm $\mu=0.66,\nu=2$ in the cubic-quintic SH equation. \label{f:almost_snake}}.
\end{figure}
In the 2D-planar cubic-quintic SH equation, stationary ``almost planar" localised structures are known to exist and snake between 4 different fold points in $\mu$ unlike the standard snaking between 2 fold points as seen in 1D; see figure~\ref{f:almost_snake} and \cite{burke2007b,avitabile2010}.  The question of how the snaking structure affects the depinning fronts of almost planar localised structures has not been explored.

Two types of bifurcations of depinning fronts one might be interested investigating are fold bifurcations and transverse instability bifurcations. 
One can readily imagine fronts ceasing to exist as the selected wavenumbers vary but it is difficult to determine from initial value solvers where these occur as opposed to the front just losing stability. 
Transverse instability bifurcations can readily be observed from time simulations; see figure~\ref{f:almost_ivp}. In the cubic-quintic SH equation with $\nu=2$,
the homoclinic snaking region in 1D (corresponding to parallel stripes) is between $0.6267<\mu<0.7126$ and an almost planar snaking region in almost the same parameter region; see~\cite{avitabile2010}. For $0.375<\mu<0.62$, we see an almost planar invasion process. However, for $0<\mu<0.375$ where the invasion speed increases, we see that the invasion front is stable with respect to transverse perturbations.  Investigating these phenomena using numerical path-following routines would allow for a systematic parameter study to be carried out that would be time consuming or impossible to carry out with time simulations alone. 
\begin{figure}[h]
	\centering
	\includegraphics[width=0.95\linewidth]{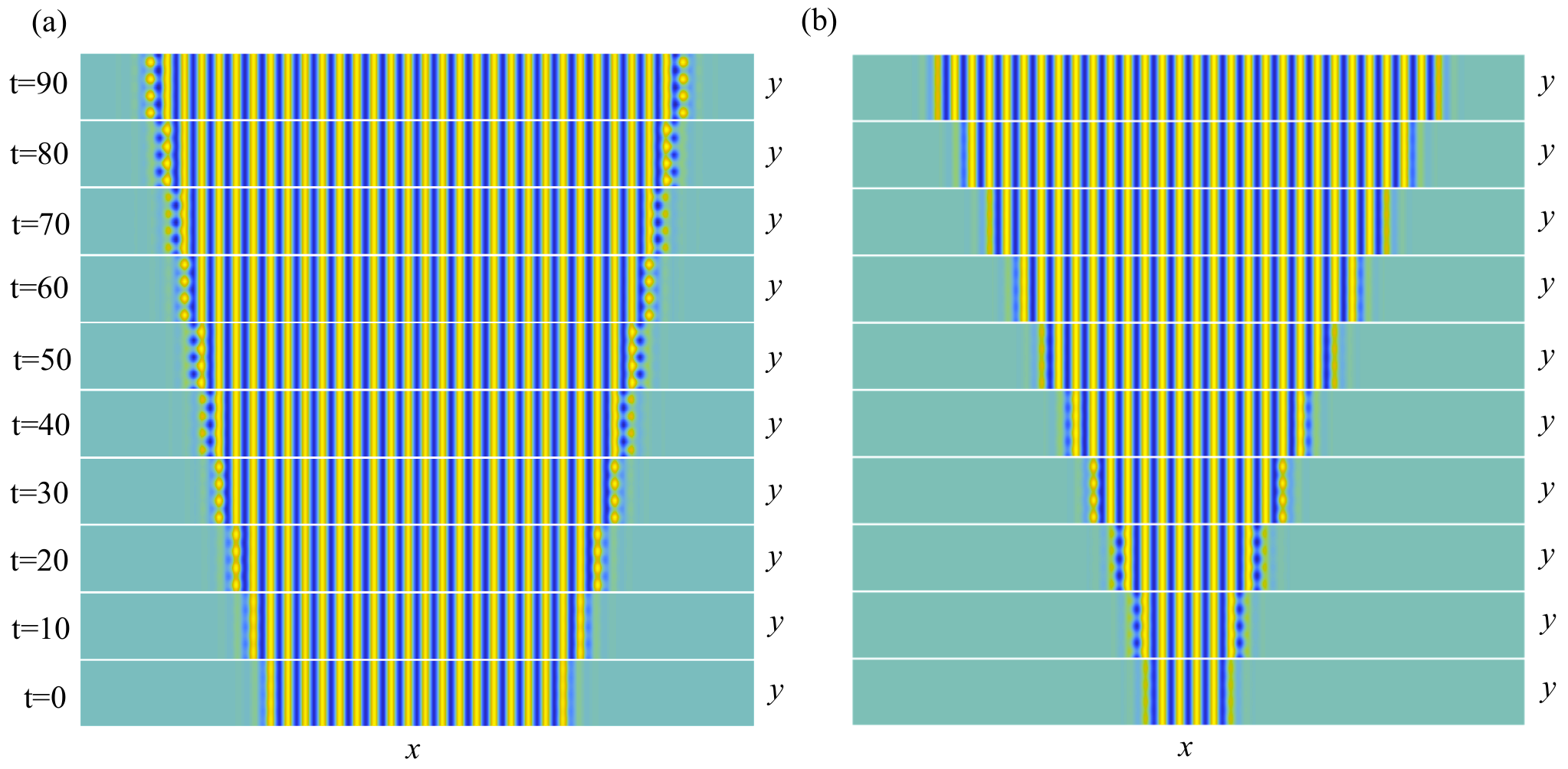}
	\caption{Time simulations of the cubic-quintic SH equation with $\nu=2$ starting from a localised stripe pattern with a small random perturbation. Panel (a) $\mu=0.45$ shows de-stabilisation to an almost planar invasion front while panel (b) $\mu=0.35$ shows an initial de-stabilisation to an almost planar invasion front before then evolving to a parallel stripe invasion front.  \label{f:almost_ivp}}
\end{figure}

Beyond investigating planar fronts, a major interest is in the depinning of 2D patches of cellular pattern; see figure~\ref{f:almost_snake}(b) for localised worms in the depinning region. In \cite{lloyd2008,avitabile2010}, they showed that the planar fronts appear to be related to the fully localised patches. The depinning of the fully localised patches has not been studied and it would also be interesting to see if the depinning planar fronts can help explain the behaviour of the patches.  Key questions one may wish to ask are: can we predict the speed of growth of the patch? and can we predict the structure of the growing patch? The stationary versions of these structures have been shown numerically to have a bifurcation structure that can be explained via planar fronts. It seems reasonable to ask if planar depinning fronts can give us an indication as to the pattern formation structure of growing 2D patches.

Pattern forming fronts propagating into an unstable state in one space dimension have been extensively studied; see~\cite{saarloos2003} for a review.  
However in the subcritical region, that is the focus of this paper, the depinning fronts are propagating into a stable state. It appears that relatively little is known numerically in the subcritical region despite the depinning fronts taking up most of the bistable region in parameter space. In doubly diffusive convection, interesting instabilities of the invasion fronts outside the pinning region have recently been investigated by Beaume {\it et al.}~\cite{beaume2018} using time simulations. There have also been recent developments in two spatial dimensions with directional quenching (i.e. the $\mu$ parameter has a travelling step change) where the fronts connect a quiescent state and a stripe patterned state; see~\cite{monteiro2017,goh2017,avery2018}. 

It is not entirely clear how to set up the numerical boundary value problem for depinning fronts for use in continuation and bifurcation routines; \cite[Chapter 10]{krauskopf2007} highlight that the computation of modulated waves is an open problem and are difficult to compute due to the presence of the essential spectrum. A natural approach is to consider depinning fronts rather than travelling invading localised pulses so that one can go to a traveling frame and use a standard phase condition for front speed~\cite{krauskopf2007}. However, one now has the problem that the far-field pattern is now slanted prohibiting simple Robin boundary conditions which standard numerical continuation packages such as \textsc{pde2path}~\cite{pde2path,uecker2014b} and \textsc{AUTO07p}~\cite{auto07p} usually require. Worse still is that the linearisation about the depinning front is not Fredholm prohibiting the application of Newton's method. To overcome this problem, we will adapt and extend and idea by Lloyd \& Scheel~\cite{lloyd2017} who proposed a new numerical method for stationary planar grain boundaries; see also Morrissey \& Scheel~\cite{morrissey2015} and Avery {\it et al.}~\cite{avery2018}. Their key idea is to use a {\it far-field core decomposition} with additional phase conditions to deal with the far-field boundary condition and setup a numerical boundary value problem for both the front and the pattern selection parameters.  This boundary value problem can then be easily embedded into numerical continuation routines to find various bifurcations. 
In this paper, we show how the far-field core decomposition idea can be extended to compute the depinning fronts and their pattern selection wavenumbers and propagation speeds.

In this paper, we concentrate on depinning fronts whose far-field is made up of stripes and leave the study of depinning fronts involving other cellular periodic structures e.g. hexagons, for a subsequent paper. We have two main results. The first of these is a pattern selection principle for transverse invasion fronts in the bistable region.
\begin{Proposition}\label{prop:1} (Pattern Selection for Parallel Depinning Fronts)
Let $\mu>0$. Assume that the far-field stripe solution is temporally stable and there exists a transverse parallel invasion front of~(\ref{e:sh}). Then invasion fronts persists under small perturbations and select a unique invasion speed, $c$ and far-field spatial wavenumber, $k_x$ for fixed $\mu$ and $\nu$. 
On the other hand, retreating fronts come in one-parameter families that are parameterised by $k_x$ and whose front speed, $c=c(k_x)$.  
\end{Proposition}
We refer the reader to Appendix~\ref{s:mod_fronts_inf} for a precise list of the assumptions and the definition of a transverse front. Transverse in this context means the kernel of the linearisation about the front is two-dimensional spanned by elements induced by spatial and temporal translations only. We believe this result is true for a more general class of PDEs. Since there is a unique pattern selection principle only for invasion fronts, we numerically concentrate on computing invading fronts. 
The proof of this result can be turned into a numerical algorithm allowing us to path follow invasion fronts for the first time. We note that a minor adaption of this result and using the results in~\cite{goh2017} yields that oblique and almost-planar fronts (see figure~\ref{f:stripe_fronts}(b) \& (d)) are also pattern selecting if invading and come in one-parameter $k_y$ families. The wavespeed selection property for perpendicular fronts (see figure~\ref{f:stripe_fronts}(b)) can easily be shown using Laypunov-Schmidt reduction. 

The second result is a comprehensive numerical bifurcation study of the invasion (depinning) fronts (parallel, perpendicular and oblique stripes, and almost planar stripes) in the planar SH equation. We will primarily focus on the quadratic-cubic SH equation for parallel invasion fronts as there is an interesting prediction from the weakly-nonlinear theory for the front propagation speed and far-field wavenumber selection. For the other fronts, we will focus on the cubic-quintic SH equation as stripes are easily destabilised in the quadratic-cubic SH equation on the plane.  Rather than give a complete listing of our results, we highlight two interesting observations.\\

{\bf Observations:}\\
{\it We observe in the bistable region:
\begin{itemize}
	\item As one moves away from the snaking region, the selected wavenumber of parallel invasion fronts starts with the Hamiltonian selected wavenumber at the edge of the homoclinic snake and displays a ``dip" before increasing. 
	\item Almost planar invasion fronts bifurcate off parallel invasion fronts and co-exist close to the homoclinic snaking region. They come in one parameter families parameterised by $k_y$. Sufficiently far from the homoclinic snaking region, no almost planar invasion fronts exist. \\
\end{itemize}
}

The paper is outlined as follows. We first briefly review the existence and stability of stripe cellular patterns on the plane in the cubic-quintic SH equation in \S\ref{s:roll_stab}. We then review pattern selection principles for stationary fronts in \S\ref{s:pat_select_stat}. In \S\ref{s:weak}, we look at the weakly nonlinear amplitude analysis for the emergence of the localised states and travelling fronts. We also review in \S\ref{s:semi_ana} the weakly nonlinear theory and semi-analytical perturbation analysis of~\cite{burke2006,aranson2000}. 
In \S\ref{s:method}, we outline the far-field core decomposition numerical method and then present our numerical investigation in \S\ref{s:numerics}. In \S\ref{s:patch_invasion}, we investigate the invasion process of stripe patches in the plane. Finally, we conclude in \S\ref{s:discussion} and outline areas for future research. 

\section{Existence and stability of stripes}\label{s:roll_stab}
In this section, we state the existence and stability results of stripes near $\mu\sim0$ in the cubic-quintic SH equation leaving the proofs of propositions 2.1 and 2.3 to Appendix~\ref{s:lay_proof}. We also numerically compute the stability boundaries of the stripes for larger values of $\mu$ beyond the validity of the existence proofs. This information will prove crucial in understanding the invasion fronts and to calculate Fredholm indices in the Appendix~\ref{s:mod_fronts_inf}. For the quadratic-cubic SH equation the results are similar except stripes are typically unstable in the bistability region to hexagonal perturbations.

We follow Mielke~\cite{mielke1995a} for the proof of existence and instability analysis for $(\mu,\nu)\sim0$. We assume that the stationary stripes are periodic in $x$ and independent of $y$. Hence, we consider solutions of the form $u(x,y,t)=\tilde u(\xi;\mu,\nu,k_x)$, where $\xi=k_xx$ and $\tilde u$ is $2\pi$-periodic in $\xi$. The symmetry group is $\Omega=\frac{2\pi}{k_x}\mathbb{Z}\times\mathbb{R}$ with discrete dual group $\Gamma=k\mathbb{Z}\times\{0\}$. We further assume, that we are close to the co-dimension 2 point where the stripes bifurcate sub/supercritically such that $\nu\sim0$. 
We also set $k_x=\sqrt{1+\kappa}$ and assume $\kappa$ is small. 

Omitting tilde's and redefining $x$, we define nonlinear problem $F(\mu,\kappa,\nu,u)=0$, where
\begin{equation}\label{e:non_F}
F(\mu,\kappa,\nu,u) := -(1+(1+\kappa)\partial_x^2)u - \mu u + \nu u^3 - u^5,\qquad F:\mathbb{R}^4\times X^4\rightarrow X^0,
\end{equation}
and $X^j = H^j_{\mbox{per}}([0,2\pi],\mathbb{R})$. We obtain the existence of stripes via Lyapunov-Schmidt reduction.

\begin{Proposition} (Existence of small amplitude stripes)
There is an $\mu_0>0$ and $\nu_0>0$ such that for all $\mu\in(0,\mu_0]$ and $\nu\in(0,\nu_0]$ all $\kappa\in\left(-\hat\kappa,\hat\kappa\right)$ where $\hat\kappa = \sqrt{\frac{9\nu^2-40\mu}{40}}$ there is a unique solution $u=\tilde u_{\mu,\kappa,\nu}$ of (\ref{e:non_F}) which is even in $x$ and positive at $x>0$. All other small bounded solutions are obtained from translates of this family. We have the expansion
\[
\tilde u_{\mu,\kappa,\nu}(\xi) = \tilde a \cos(\xi) + \mathcal{O}(\tilde a^2),
\]
where
\[
\tilde a(\mu,\kappa,\nu)^2_\pm = \left(\frac{1}{2\pi}\langle \tilde u_{\mu,\kappa,\tilde\nu},U_1\rangle \right)^2= \frac{3\nu}{5}\pm\frac{\sqrt{9\nu^2-40(\mu+\kappa^2)}}{5} + \mathcal{O}(|\mu| + |\nu|^2 +|\kappa|^2).
\]
\end{Proposition}
We note that we require $9\nu^2\geq40(\mu+\kappa^2)$ where we have four branches. Folds on the branches occur when $9\nu^2=40(\mu+\kappa^2)$. It is clear, the last stripe to lose existence is the $\kappa=0$ one at $\mu=9\nu^2/40$ and a parabola existence region of stripes grows as $\mu$ is decreased; see figure~\ref{f:stripe_stab}. We call $\tilde a_-$ the ``lower" branch and $\tilde a_+$ the upper branch in the bistable region.

Linear stability of the periodic solutions to (\ref{e:non_F}) can be found by looking at the linear problem
\[
-(1+\partial_x^2 + \partial_y^2)^2w + \mu w + (3\tilde\nu\epsilon^2\tilde u_{\mu,k,\tilde\nu}^2(x) - 5\tilde u_{\mu,\kappa,\tilde\nu}^4(x))w = \lambda w.
\]
Carrying out a Floquet-Bloch decomposition of the form $w = e^{i(\sigma x + \tau y)}\tilde w(\xi)$ where $\xi = k_xx$ and $\tilde w\in X^4$, leads to the eigenvalue problem
\begin{equation}\label{e:roll_lin_op}
L(\mu,k,\nu,\sigma,\tau,\lambda)\tilde w := -(1+(k_x\partial_\xi + i\sigma)^2 - \tau^2)^2\tilde w + \mu \tilde w+ (3\tilde\nu\epsilon^2\tilde u_{\mu,\kappa,\nu}^2(x) - 5\tilde u_{\mu,\kappa,\nu}^4(x))\tilde w  - \lambda \tilde w = 0
\end{equation}

\begin{Lemma}~\cite{mielke1995a} \label{l:stab_rolls}(Quadratic dependence on $\sigma$ and $\tau$ of the zero eigenvalue)

The critical eigenvalue $\lambda\sim 0$ has the expansion
\begin{align*}
\lambda =& c_1(\mu,\nu,k_x)\sigma^2 + c_2(\mu,\nu,k_x)\tau^2 + \mathcal{O}(\sigma^4+\tau^4),\qquad\mbox{where},\\
c_1(\mu,\nu,k_x) =& \frac{2}{c(\mu,\nu,k_x)}\left[8k_x^2(1-k_x^2)^2 -(1+k_x^2)c(\mu,\nu,k_x) \right] + \mathcal{O}(\tilde a^2), \quad\mbox{and},\\
\quad c_2(\mu,\nu,k_x) =& \frac{2(1-k_x^2)}{c(\mu,\nu,k_x)}+\mathcal{O}(\tilde a^2),
\end{align*}
where $c(\mu,\nu,k_x) = \frac34\nu\tilde a^2 - \frac54\tilde a^4+\mathcal{O}(\tilde a^6)$.
Instability holds when ever $c_1$ or $c_2$ are positive. 
\end{Lemma}
\begin{Proof}
The proof is identical to that of~\cite[Lemma 4.2]{mielke1995a}.
\end{Proof}
It is easy to see that the lower branch, $\tilde a_-$, is always unstable to co-periodic perturbations. The quadratic dependence of the critical eigenvalue on $\sigma$ and $\tau$ will form a key part in showing the pattern selection mechanism of the invasion fronts in \S\ref{s:mod_fronts_inf}.

 In order to determine the most unstable perturbations, we introduce the scalings
 \[
\mu = \epsilon^4 \hat\mu,\qquad\nu = \epsilon^2\hat\nu,\qquad (k_x^2-1)=2\epsilon^2\hat\kappa,\qquad\sigma = \epsilon^2\hat\sigma,\qquad\tau = \epsilon\hat\tau,\qquad\lambda = \epsilon^4\hat\lambda,
 \]
 where $\hat\kappa \in[-\frac{\tilde\kappa}{2},\frac{\tilde\kappa}{2}]$, $|\epsilon|\ll1$.
  
 \begin{Proposition}\label{p:stripe_instab}(Instability of small amplitude stripes) The lower branch of stripes $\tilde a_-$ is always unstable. For the upper branch of stripes $\tilde a_+$, we have three curves
 \[
 K^{\mbox{z}} = 1+ \mathcal{O}(\epsilon^8), \qquad\mbox{and},\; K^{\mbox{e}}_\pm = 1\pm\frac{135\hat\nu^2-800\hat\mu+15\hat\nu\sqrt{81\hat\nu^2-320\hat\mu}}{6400}\epsilon^2 + \mathcal{O}(\epsilon^4)
 \]
 such that the stripes are Eckhaus unstable for $k_x\not\in[K^{\mbox{e}}_-(\epsilon,\hat\mu,\hat\nu),K^{\mbox{e}}_+(\epsilon,\hat\mu,\hat\nu)]$ and zig-zag unstable for $k_x<K^{\mbox{z}}(\epsilon,\hat\mu,\hat\nu)$
 \end{Proposition}
 The existence and stability boundaries of the upper branch of stripes in the bistable region look very similar to that for the SH equation with just a re-stabilising cubic nonlinearity and we do not plot them here. 
 We note that while this is only an instability result, numerical calculations below suggest that the stripes are stable for $k_x\in[K^z(\epsilon,\nu),K_+^e(\epsilon,\nu)]$. A more detailed analysis along the lines of~\cite{mielke1997} would be required to establish stability and we do not carry it out here.  
 
 \begin{figure}[h]
	\centering
	\includegraphics[width=0.8\linewidth]{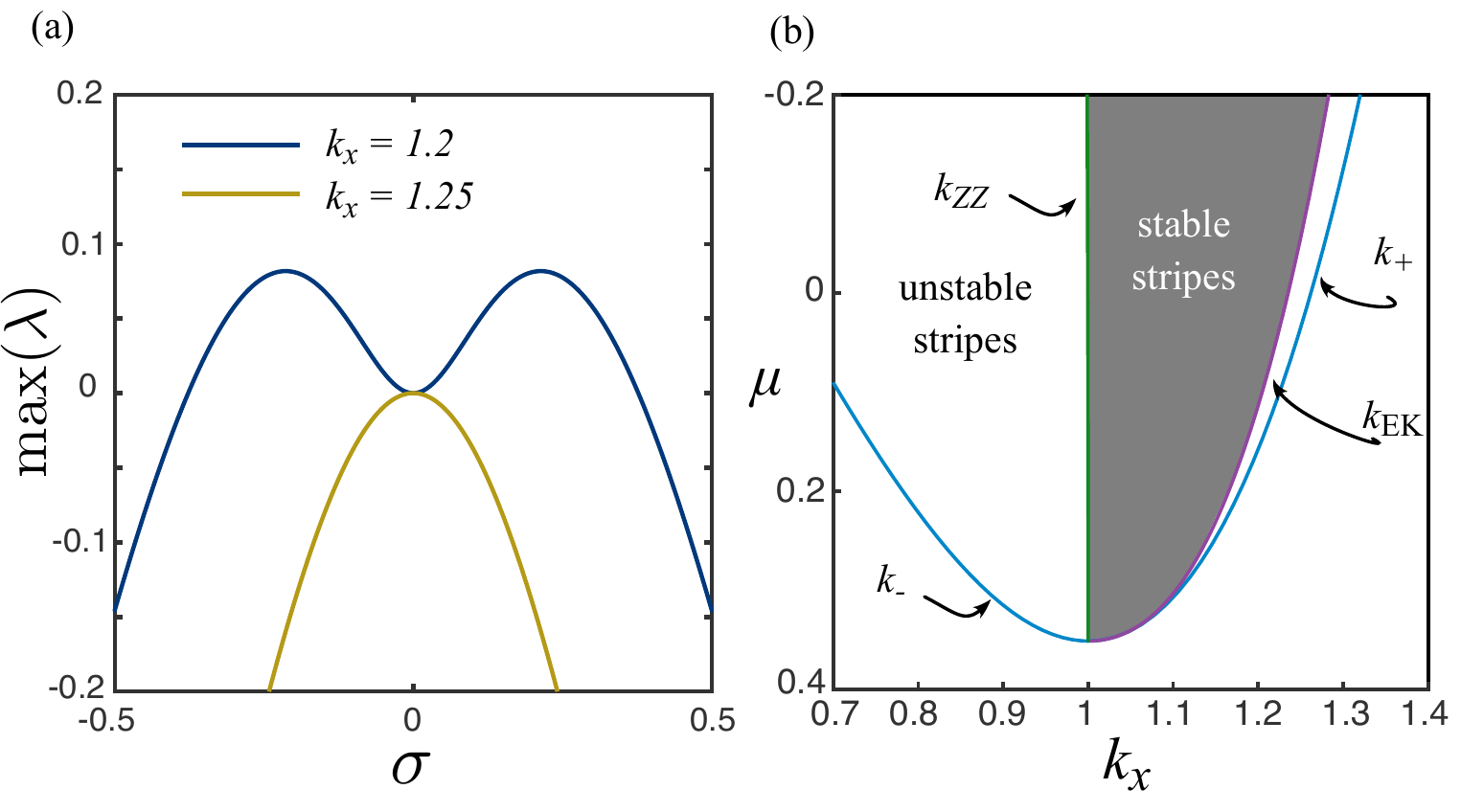}
	\caption{(a) $\max(\lambda)$-eigenvalue of the periodic orbit as the Floquet multiplier $\sigma$ is varied for $(\mu,\nu)=(0,1.25)$. All other eigenvalues are strictly negative for all $\sigma$. For $k_x=1.2$ the stripe is Eckhaus stable and unstable for $k_x>1.2365$. (b) Stripe existence and 2D stability curves for $\nu=1.25$. $k_{\pm}$ denotes the existence boundaries for the stripes, $k_{\mbox{EK}}$ denotes the Eckhaus instability boundary, and $k_{\mbox{ZZ}}$ the zig-zag instability boundary. \label{f:stripe_stab}}
\end{figure}
 
Far away from the weakly nonlinear limit, we can numerically trace out the existence and stability boundaries. The existence boundaries are found by continuing in $k_x$ and path following in two parameters the fold loci. The zig-zag instability eigenvalue, $\lambda_{\mbox{zz}}$, is found by computing
\[
\lambda_{\mbox{zz}} = 2\frac{\langle(\partial_x^2-1)u_x,u_x \rangle}{\langle u_x,u_x\rangle},
\]
for a 1D stripe; see for instance~\cite{lloyd2017}. The boundary is then found by setting $\lambda_{\mbox{zz}}=0$ and tracing out the resulting curve $k_x=k_{\mbox{ZZ}}(\mu)$. The Eckhaus instability boundary is found by computing when the zero eigenvalue changes sign in its curvature as the Floquet multiplier $\sigma$ is varied; see~\cite{radss,sherratt2012}. In particular, we solve for $(u,v,v_{\sigma},v_{\sigma\sigma},c,k_x,\lambda,\lambda_{\sigma})$ and continue the following system
\begin{align*}
-(1+k_x^2\partial_x^2)^2u - \mu u + f(u) + cu_x =& 0,\qquad x\in(0,2\pi], c\in\mathbb{R},\\
Lv - \lambda v =& 0,\qquad v,\lambda\in\mathbb{R},\\
Lv_{\sigma} - \lambda_{\sigma}v - 4(1+k_x^2\partial_x^2)k_x\partial_xv =&0, \qquad v_{\sigma},\lambda_{\sigma}\in\mathbb{R},\\
Lv_{\sigma\sigma} + 2\lambda_{\sigma}v_{\sigma}  + 8(1+k_x^2\partial_x^2)k_x\partial_xv_{\sigma} + 4v + 12k_x^2\partial_x^2v =&0,\qquad v_{\sigma\sigma}\in\mathbb{R},\\
\int_0^{2\pi}u_x^{\mbox{old}}(u-u^{\mbox{old}})\D x = 0,\qquad 
\int_0^{2\pi}v^2\D x =& 1,\\
\int_0^{2\pi}vv_{\sigma}\D x = 0,\qquad
\int_0^{2\pi}vv_{\sigma\sigma}\D x =& 0,
\end{align*}
where $L = -(1+k_x^2\partial_x^2)^2 - \mu + f'(u)$ and $u^{\mbox{old}}=\cos(x)$. This system is found by differentiating (\ref{e:roll_lin_op}), with respect to $\sigma$. We provide an initial condition 
\[
(u,v,v_{\sigma},v_{\sigma\sigma},c,k_x,\lambda,\lambda_{\sigma})=(u_s,(u_s)_x,0,0,0,0,0,0),
\]
where $u_s=u_s(x)$ is the converged periodic orbit, to then use Newton's method to converge to a solution.

In figure~\ref{f:stripe_stab}(a), we plot $\lambda(\sigma)$ with $\tau=0$ for $\nu=1.25$ with $k_x=1.2$ and $k_x=1.25$ shown in blue and gold, respectively. We see that for sufficiently large $k_x$, there is a change in curvature of the eigenvalue $\lambda(\sigma)$ leading to the Eckhaus instability shown in blue. 

In figure~\ref{f:stripe_stab}(b), we plot a typical existence and stability diagram for the upper branch of stripes at $\nu=1.25$. Here we see that qualitatively the stability boundaries are rather similar to those observed in the weakly nonlinear limit. In particular, we see that the zig-zag instability is very close to unity as predicted in Proposition~\ref{p:stripe_instab} and the Eckhaus instability curve is quadratic in the wavenumber, $k_x$.

\section{Pattern Selection Principles for Stationary fronts}\label{s:pat_select_stat}
In this section, we review the pattern selection principles for stationary pattern fronts connecting to the trivial state in the snaking region. Since the stripe patterns come in families parameterised by their wavenumbers/periods, a natural question arises about if a front selects a unique wavenumber and the mechanism for the selection. When investigating invasion fronts near the edge of the homoclinic snaking region we will be interested in how their pattern selection relates to that of the stationary fronts. 

In 1D, this pattern selection mechanism of stationary fronts connecting periodic stripes is completely resolved for the SH equation. Rewriting the stationary SH equation as a first order ODE system $U_x = F(U)$, where $U:=(u_1,u_2,u_3,u_4)^T= (u,u_x,u_{xx},u_{xxx})^T$ and $F(U) = (u_2,u_3,u_4,-2u_2 -(1+\mu)u_1 + f(u_1))^T$, one can show the system has the Hamiltonian 
\begin{equation}\label{e:1D_Ham}
\mathcal{H}_{\mbox{1D}}(u) = u_{xxx}u_x  - \frac12u_{xx}^2 + u_x^2 + \frac12(1+\mu)u^2 + G(u),
\end{equation}
where $G(u) = \int_uf(u)\D u$, which is independent of $x$. Hence, if we have a 1D front connecting a stripe state to the trivial state, then $\mathcal{H}_{\mbox{1D}}$ must vanish when evaluated along a single stripe in the far field of the front. A minor modification of~\cite[Proposition 3]{lloyd2008} shows that there exist stripe patterns that satisfy the Hamiltonian selection principle. 

One can also use the algorithms outlined in \S\ref{s:method} with $\omega=0$ or in~\cite{krauskopf2008} to compute stationary fronts connecting a spatially periodic pattern to the trivial state, but we do not do this here. Instead we plot for the sake of completeness the Hamiltonian selected wavenumber of the 1D stripes for both the quadratic-cubic and cubic-quintic SH equation in figure~\ref{f:1D_ham}. The bifurcating selected stripes have an initial wavenumber of one and are unstable until a fold where they subsequently restabilise after which, the selected wavenumber monotonically decreases.

\begin{figure}[h]
	\centering
	\includegraphics[width=0.8\linewidth]{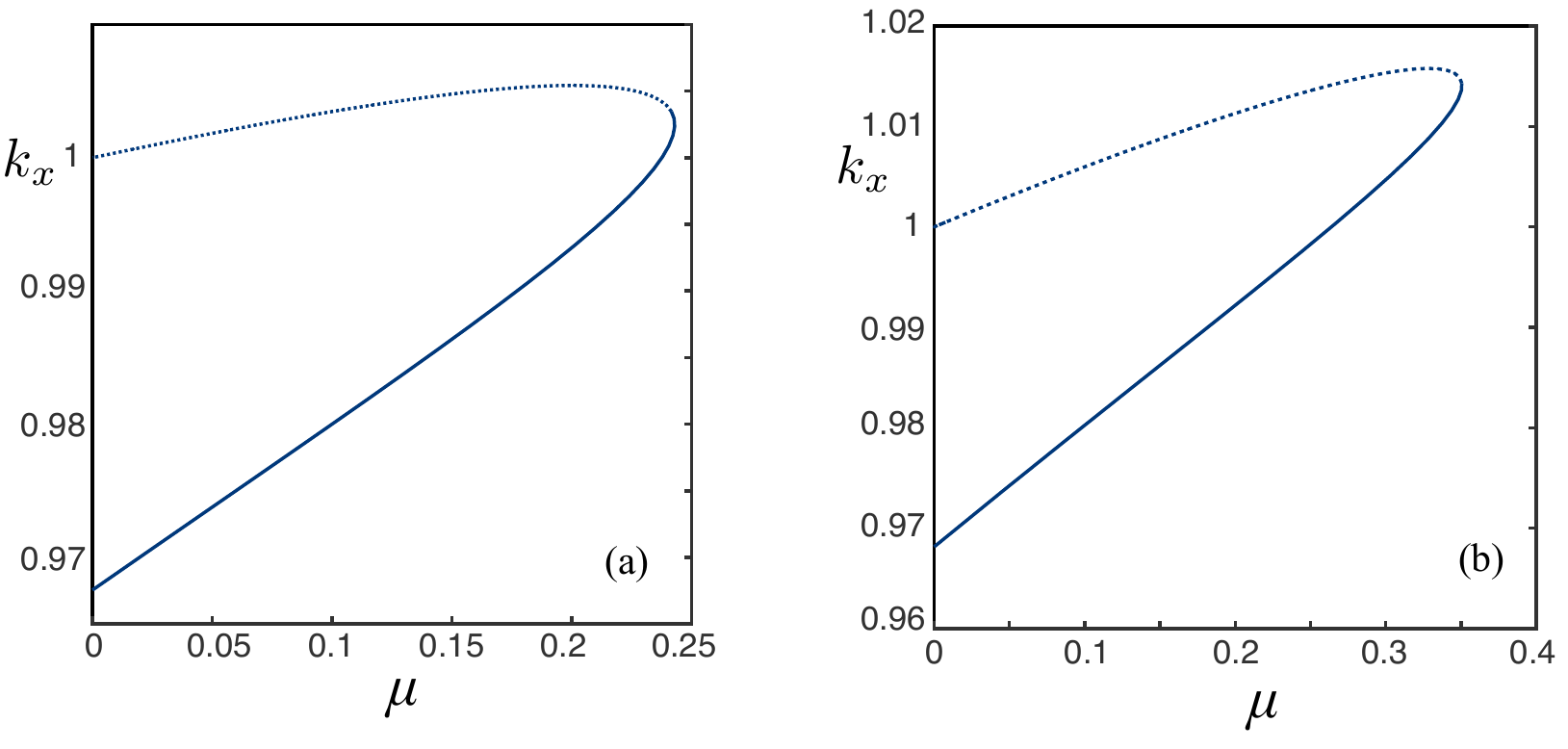}
	\caption{Hamiltonian selected wavenumber for (a) 1D stripes in the quadratic-cubic SH equation $\nu=1.6$, (b) 1D stripes in the cubic-quintic SH equation $\nu=1.25$. The co-periodic stable branches are depicted as solid lines while co-periodic unstable branches are shown as dashed lines.\label{f:1D_ham}}
\end{figure}

In addition to the Hamiltonian, the SH equation posed on $\mathbb{R}^d$ with $1\leq d\leq 3$ is a gradient system 
\[
u_t = -\nabla\mathcal{E}(u),
\]
in $H^2(\mathbb{R}^d)$, where the energy functional, $\mathcal{E}$, is given by 
\[
\mathcal{E}(u) = \int_{\mathbb{R}^d}\left[\frac{[(1+\Delta)u]^2}{2} + \frac{\mu u^2}{2} + G(u)\right] \D\mathbf{x},\qquad \mathbf{x}\in\mathbb{R}^d,
\]
and the gradient $\nabla\mathcal{E}(u) = \frac{\delta\mathcal{E}}{\delta u}(u)$ of $\mathcal{E}$ with respect to $u$ is computed in $L^2(\mathbb{R}^d)$. One often looks for points in parameter space where the energy of a single stripe, that also lies in the zero-level Hamiltonian set, has zero energy. This point in parameter space is known as the {\it Maxwell point}. At the Maxwell point, one expects to find a stationary front connecting the trivial state to stripes that have zero Hamiltonian. The homoclinic snaking is found to occur around the Maxwell point. 

In 2D, the pattern selection mechanism is similar except we now have two conserved quantities that are independent of $x$. 
\begin{Proposition}~\cite{lloyd2008,lloyd2017} (Conserved quantities for the 2D SH Equation)

If $u(x,y)$ is a smooth stationary solution of the planar SH equation~(\ref{e:sh}) which is spatially periodic with period $\ell$ in the $y$-variable, then the quantities 
\begin{align*}
\mathcal{H}(u) =& \int_0^\ell \left[u_{xxx}u_x  - \frac{u_{xx}^2}{2} + u_x^2 + \frac{(1+\mu)u^2}{2} + G(u) - u_{xy}^2 - u_y^2 + \frac{u_{yy}^2}{2} \right]\D y,\\
\mathcal{S}(u) =& \int_0^\ell\left[(u + u_{xx} + u_{yy})u_{xy} - u_y(u_x + u_{xxx} + u_{xyy}) \right] \D y,
\end{align*}
where $G(u) = \int_0^uf(v)\D v$, do not depend on $x$.
\end{Proposition}
The first conserved quantity was proven in~\cite{lloyd2008} while the second quantity was recently pointed out by~\cite{lloyd2017}. Both quantities can be seen from the conservation laws~\cite[Lemma 1]{lloyd2008} due to translations in $x$ and $y$, respectively,
\begin{subequations}\label{e:con_quant}
\begin{align}
\partial_x\mathcal{L}(u,\nabla u,\Delta u) - \nabla\cdot\left[u_x\mathcal{L}_p(u,\nabla u,\Delta u) + \mathcal{L}_r(u,\nabla u,\Delta u)\nabla u_x - u_x\nabla\mathcal{L}_r(u,\nabla u,\Delta u) \right] =& 0,\\
\partial_y\mathcal{L}(u,\nabla u,\Delta u) - \nabla\cdot\left[u_y\mathcal{L}_p(u,\nabla u,\Delta u) + \mathcal{L}_r(u,\nabla u,\Delta u)\nabla u_y - u_y\nabla\mathcal{L}_r(u,\nabla u,\Delta u) \right] =& 0,
\end{align}
\end{subequations}
where $\mathcal{L}(q,p,r):\mathbb{R}\times\mathbb{R}^2\times\mathbb{R}\rightarrow\mathbb{R}$ and 
\[
\mathcal{L}(q,p,r) = \frac{(q+r)^2}{2} + \frac{\mu q^2}{2} + G(q).
\]
Integrating both conservation laws in $y$ over $[0,\ell]$ and using periodicity in $y$, we find the two quantities. 

As in 1D, if we have a planar front connecting a spatially doubly periodic state (periodic in the transverse direction) to the trivial state, then both $\mathcal{H}$ and $\mathcal{S}$ must vanish when evaluated along a single periodic cell in the far field of the front. Since we have two conditions, we expect both the wavenumbers in $x$, $k_x$, and $y$, $k_y$, to be selected. 
However, for any even function in $y$, the integrand in $\mathcal{S}$ is odd and so $\mathcal{S}=0$. Hence, we do not expect pattern selection of $k_y$ for perpendicular, parallel or almost-planar stripe fronts. 

For oblique stripe fronts, there is no even symmetry in $y$ and so we may expect some sort of pattern selection to occur. In particular, for an oblique stripe $u = u_s(k_xx+k_yy;|k|)= u_s(\xi;|k|)$, where $\xi\in(0,2\pi], k = \sqrt{k_x^2+k_y^2}$, we find that $S(u)$ reduces to 
\[
\mathcal{S}(k) = 2k_xk_y\int_0^{2\pi}\left[|k|^2(u_s)'' - (u_s')^2\right]\D \xi.
\]
For marginally zig-zag stable stripes i.e., $k=k_{ZZ}$, $S(u_s)$ is zero; see~\cite{lloyd2017}. Furthermore, at $k=k_{ZZ}$, $\mathcal{H}(u_s)$ becomes
\[
\mathcal{H}(k) = \int_0^{2\pi}\left[-\frac12|k|^4(u_s'')^2 +  \frac{(1+\mu)u^2}{2} + G(u) \right]\D \xi.
\]
Now $\mathcal{H}(k)=0$ occurs when $\mathcal{E}(u_s)$ is zero i.e., at the Maxwell Point of the stripes. Hence, snaking of oblique stripe fronts is not possible and they can only exist at the Maxwell point. Since $\mathcal{S}$ and $\mathcal{H}$ only depend on $|k|$, then at the Maxwell point, any rotation of the stripes is possible. Numerically, we have confirmed that oblique stationary stripe fronts exist at the Maxwell point using the algorithms described in  \S\ref{s:method} but we do not show this here. 

\begin{Proposition}
Stationary oblique stripe fronts of the planar SH equation~(\ref{e:sh}) on a fixed periodic infinite strip $(x,y)\in\mathbb{R}\times \mathbb{S}^1$, where $\mathbb{S}^1=\mathbb{R}/2L_y\mathbb{Z}$ with period $2L_y$ in the $y$-direction, can only exist at the Maxwell point i.e., oblique stripe fronts cannot snake. At the Maxwell point, any orientation of the stripes with respect to the front interface is possible and the selected stripe wavenumber is $k_{ZZ}$.
\end{Proposition}


\section{Weakly nonlinear analysis of fronts}\label{s:weak}
In this section, we review the formal weakly-nonlinear analysis for depinning fronts near $\mu\sim0$. Formal derivation and analysis of the amplitude equations for 1D modulated fronts in the subcritical region for the generalised SH equation was first done by Kao \& Knobloch~\cite{kao2013}. 
Rigorous derivation and analysis for fronts propagating into unstable states  in the 1D supercritical cubic SH equation was carried out by Eckmann \& Wayne~\cite{eckmann1991} and extended to hexagon modulated fronts by Doelman {\it et al.}~\cite{doelman2003}.

\subsection{Parallel stripe fronts}
Since parallel stripe fronts are just the simple 2D extension of the 1D fronts, we consider the 1D SH equation with a general nonlinearity given by
\[
u_t = -(1+\partial_x^2)^2u -\mu u + f_2u^2 + f_3u^3 + f_4u^4 + f_5u^5 + \mathcal{O}(u^6).
\]
We are interested in modulated fronts near the co-dimension two point where the stripes bifurcate sub/supercritically. 
Hence, following~\cite{kao2013}, we scale the parameters 
\[
\mu = \epsilon^4\tilde \mu,\qquad f_3 = -\frac{38}{27}f_2^2 + \epsilon^2 b,
\]
where $\epsilon\ll1,b\in\mathbb{R}$, 
and we consider the ansatz
\begin{equation}\label{e:1D_ansatz}
u(x,t) = \epsilon A(X,T)e^{ix} + c.c + \mathcal{O}(\epsilon^2),
\end{equation}
where $X=\epsilon^2 x, T = \epsilon^4t$. Substituting the ansatz~(\ref{e:1D_ansatz}) into~(\ref{e:sh}) and equating at orders $\mathcal{O}(\epsilon^j),j=1,2,3,4,5$, (see~\cite{kao2013,budd2005}) one ends, up after a lot of algebra, with the amplitude equation for $A(X,T)$ at $\mathcal{O}(\epsilon^5)$ given by
\begin{equation}\label{e:amp1}
A_T = 4A_{XX} - \tilde \mu A + \frac{32if_2^2}{27}|A|^2A_X + 3b|A|^2A - 2\left(\frac{1960f_2^4}{81}- \frac{58f_2f_4}{3} - 5f_5 \right)|A|^4A.
\end{equation}
Various authors have studied stationary localised solutions of this equation; see for instance~\cite{budd2005,burke2007c,dawes2010} and figure~\ref{f:amp_stat}. The typical bifurcation structure is a subcritical bifurcation at $\tilde\mu=0$ where two localised pulses bifurcate off the trivial state; see figure~\ref{f:amp_stat}(a). As $\tilde\mu$ is increased the pulses broaden until they resemble a combination of a front and back; see figure~\ref{f:amp_stat}(b). The bifurcation curve asymptotes at the Maxwell point for the system and there is no snaking of the pulse solutions or pinning of the fronts in this equation. We note that in the full SH equation, exponentially small terms create a pinning region where homoclinic snaking takes place about the Maxwell point; see  Kozyreff \& Chapman~\cite{kozyreff2006,kozyreff2009,kozyreff2013}.
\begin{figure}[h]
	\centering
	\includegraphics[width=0.7\linewidth]{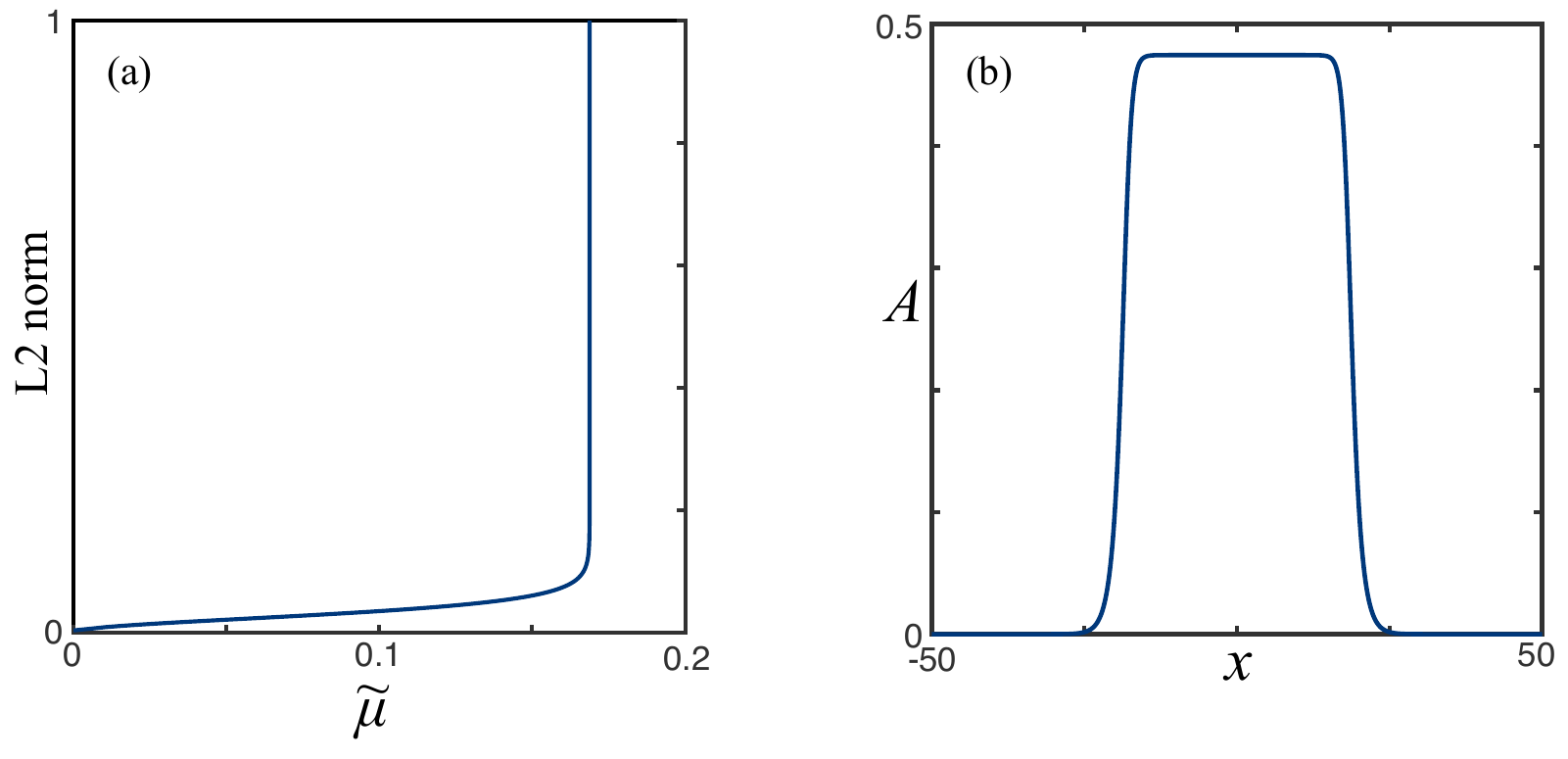}
	\caption{(a) Bifurcation diagram for a stationary pulse of~(\ref{e:amp1}) with $f_2=f_5=0,b=1$. The bifurcation diagram asymptotes at the Maxwell point (b) plot of $\Re(A)$ near the Maxwell point.\label{f:amp_stat}}
\end{figure}

The modulated fronts that we seek correspond to front solutions $A(X,T) = A(X-\tilde cT)=:A(Z)$ of (\ref{e:amp1}) and satisfy the ODE
\begin{equation}\label{e:GL_1D}
4A_{ZZ} + \tilde cA_Z - \tilde\mu A + \frac{32if_2^2}{27}|A|^2A_Z + 3b|A|^2A - 2\left(\frac{1960f_2^4}{81}- \frac{58f_2f_4}{3} - 5f_5 \right)|A|^4A = 0,
\end{equation}
where the wavespeed $c=\epsilon^2\tilde c$. 

Rescaling~(\ref{e:GL_1D}) yields the Ginzburg-Landau equation analysed by Kao \& Knobloch~\cite{kao2013} and Ponedel {\it et al.}~\cite{ponedel2017}:
\begin{equation}\label{e:rescale_GL_1D}
A_{ZZ}  + \tilde cA_Z - \tilde\mu A + ai|A|^2A_Z + b|A|^2A - |A|^4A = 0,
\end{equation}
where 
\[
a = \frac{4f_2^2}{21\sqrt{5}}\left(f_2^4 - \frac{783}{980}f_2f_4 - \frac{81}{392}f_5 \right)^{-1/2},
\]
provided $f_2^4 - \frac{783}{980}f_2f_4 - \frac{81}{392}f_5>0$. We note that in the cubic-quintic case $a=0$. The depinning fronts we are interested in correspond to nonlinear travelling fronts of the form~\cite{ponedel2017} 
\begin{equation}\label{e:CGL_front}
A(Z) = a_Ne^{iq_NZ}(1+e^{2a_N^2e_1Z})^{-\frac12-i\frac{e_0}{2e_1}},
\end{equation}
where 
\begin{align*}
a_N^2 =& \frac{2(5\Lambda - 6) + 2\Upsilon\sqrt{(2\Lambda+\tilde\mu\Delta)/\Gamma} }{\Delta},\\
q_N^2 =&  \frac{a}{\Delta}\left[-2\Lambda + (6-\Lambda)\sqrt{(2\Lambda+\tilde\mu\Delta)/\Gamma} \right],\\
\tilde c         =& \sqrt{\frac{\Gamma}{3}}\left[\frac{(\Lambda-6)+\sqrt{(2\Lambda+\tilde\mu\Delta)\Gamma}}{\Delta} \right],
\end{align*}
where $\Gamma=16-3a^2,\Delta = 16-4a^2,\Upsilon=8-3a^2,\Lambda=2,e_0=-\frac14a,e_1=\frac14\sqrt{\Gamma/3}$ and $\omega=-q_N\tilde c$, in the case where $b=1$. We call $\omega$ the {\it transition frequency}. The front connects the rotating wave state $A=a_Ne^{i(q_NZ - \omega t)}$ as $Z\rightarrow-\infty$ to the trivial state $A=0$ as $Z\rightarrow\infty$. The wavenumber of the selected stripes is given by
\[
k = 1 + \epsilon^2q_N + \mathcal{O}(\epsilon^4),
\]
and for the cubic-quintic SH equation we find the correction is higher order while for the quadratic-cubic SH equation we find the wavenumber selection is less than one. 

\begin{figure}[h]
	\centering
	\includegraphics[width=1\linewidth]{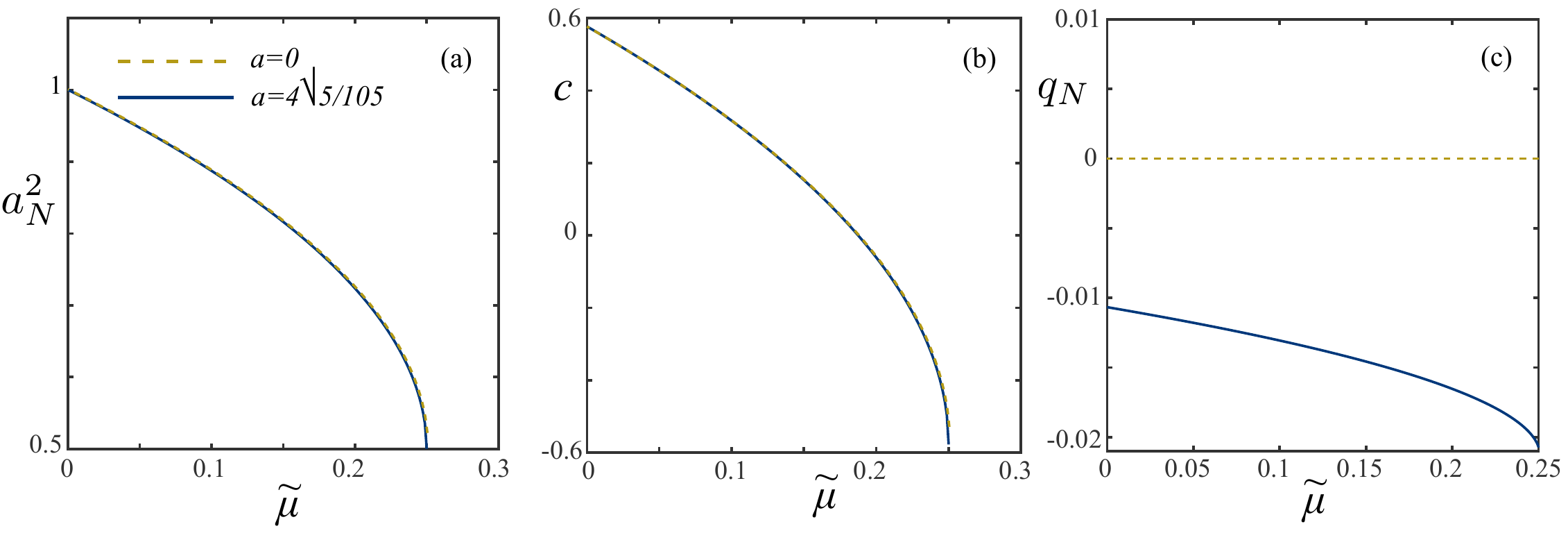}
	\caption{Bifurcation diagrams for traveling wave solution CGL~(\ref{e:rescale_GL_1D}) for the quadratic-cubic and cubic-quintic SH equation with $b=1$. (a) shows the square of the amplitude of the travelling wave (b) the selected wavespeed and (c) the selected far-field wavenumber. We see that for the quadratic-cubic SH equation with $a\neq0$, the selected far-field wavenumber monotonically decreases as one approaches the fold of the stripes in the bistable region.  \label{f:parallel_amp}}
\end{figure}

In figure~\ref{f:parallel_amp}, we plot the amplitude, $a_N$, selected correction wavenumber, $q_N$, and the wavespeed for the quadratic-cubic (with $a=4\sqrt{5}/105$) and cubic-quintic SH equation, both with $b=1$. The main observation is that as we go further in to the bistability region, both the amplitude and selected front speed decrease monotonically. Furthermore, for the quadratic-cubic SH equation, the amplitude equations predict that the selected wave number should also decay monotonically in particular there is also a unique pattern selection mechanism for retreating fronts.

We note that all of this can be put on a rigorous foundation using the spatial dynamics ideas in~~\cite{haragus1999} and transversality of the fronts in the extended phase space (including the wavespeed $c$) has been proven in~\cite{duan1995}.

\subsection{Oblique and perpendicular stripe fronts}
It is possible to capture oblique invasion stripes in figure~\ref{f:stripe_fronts}(b) for $\mu\sim0$ using weakly nonlinear amplitude equations~\cite{malomed1990}. To do this we consider the ansatz
\begin{equation}\label{e:ob_ansatz}
u(x,y,t) = \epsilon A(\epsilon^2x,\epsilon^4t)e^{ikr} + c.c. + \mathcal{O}(\epsilon^2),
\end{equation}
where $r = \cos(\alpha)x - \sin(\alpha)y$ and $\alpha$ is the orientation of the interface with respect to the strip pattern with $\alpha=0$ corresponding to parallel stripes. We concentrate on just the cubic-quintic SH equation. Upon substituting in~(\ref{e:ob_ansatz}) into the SH equation with 
\[
\mu = \epsilon^4\tilde\mu,\qquad \nu = \epsilon^2\tilde\nu,
\]
at $\mathcal{O}(\epsilon^5)$ the amplitude equation for $A$ is given by
\begin{equation}
A_T = 4(\cos(\alpha))^2A_{XX} - \tilde\mu A  + 3\tilde\nu |A|^2A - 10|A|^4A.
\end{equation}
We note that at $\alpha=\pi/2$, the $A_{XX}$ term disappears and this corresponds to the stripes becoming perpendicular. Going to a travelling frame $Z=(X-\hat c T)/\cos(\alpha)$ where $A(X,T)=A(Z)$, we find the equation 
\[
4A_{ZZ}+\frac{\hat c}{\cos(\alpha)}A_Z - \tilde\mu A + 3\tilde\nu|A|^2A - 10|A|^4A = 0.
\]
This is the same equation as~(\ref{e:GL_1D}) just with a different wavespeed i.e. $\tilde c_{1D} = \hat c/\cos(\alpha)$. Hence, we have the same traveling front solution but with a wavespeed $\tilde c/\cos(\alpha)$. Therefore, weakly nonlinear oblique stripe fronts are predicted to travel slower than parallel (1D) depinning fronts is slower than the parallel fronts in the absence of pinning effects/exponentially small corrections. 

\begin{figure}[h]
	\centering
	\includegraphics[width=0.8\linewidth]{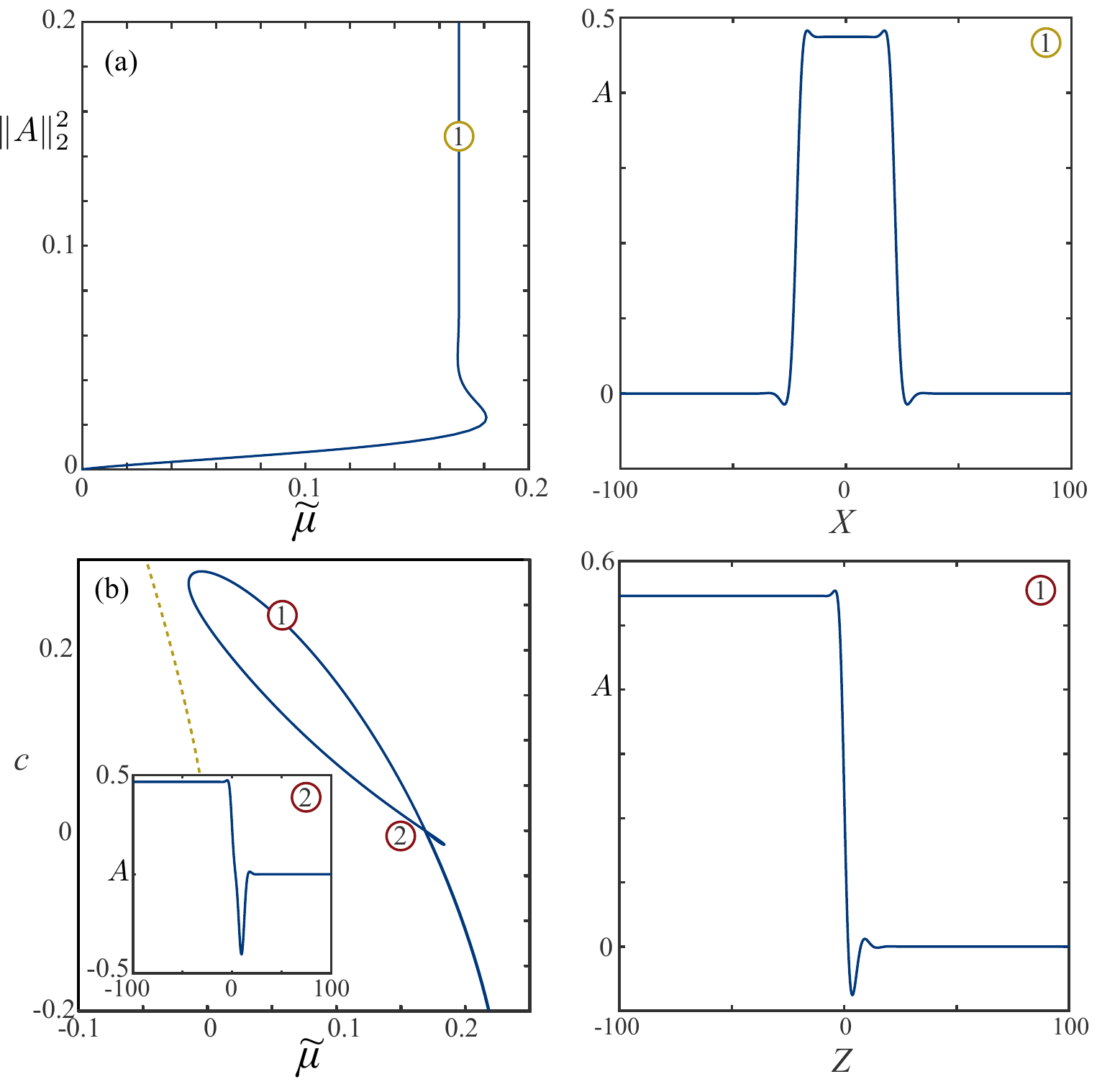}
	\caption{Bifurcation diagram for localised patterns in the perpendicular amplitude equation~(\ref{e:perp_amp}) with $\tilde\nu=1$ (a) stationary pulse that bifurcates from the trivial state (b) travelling fronts in the frame $Z=X-\tilde cT$. In panel (b) we also plot in dashed gold the linear spreading speed. \label{f:perp_amp}}
\end{figure}

When $\alpha=\pi/2$ the stripes become perpendicular and we change the scaling in the ansatz 
\[
u(x,y,t) =  \epsilon A(\epsilon x,\epsilon^4t)e^{iy} + c.c. + \mathcal{O}(\epsilon^2).
\]
The corresponding amplitude equation at $\mathcal{O}(\epsilon^5)$ is given by
\begin{equation}\label{e:perp_amp}
A_T = -A_{XXXX} - \tilde\mu A + 3\tilde \nu |A|^2 A - 10|A|^4A.
\end{equation}
This equation is still very difficult to analyse. The existence of stationary pulses can be proved using topological shooting methods (see Peleiter \& Troy~\cite{peletier2001}) and these pulses broaden until they reach the 1D Maxwell point much like the standard bifurcation diagram seen for the amplitude equation with just a second-order derivative in $X$; see Figure~\ref{f:amp_stat} and \ref{f:perp_amp}(a). However, the existence of travelling fronts looks difficult to establish. Instead, we numerically trace out the travelling fronts in figure~\ref{f:perp_amp}(b) in the travelling frame $Z=X-\tilde cT$. Here we see a significant difference for fronts compared to the parallel stripe case. The stationary front that the stationary pulses appear to converge at the Maxwell point, start to travel as we move away from the Maxwell point. As $\tilde\mu$ is decreased, the travelling fronts undergo a fold bifurcation and loop back towards the Maxwell point. The location of the fold depends on $\tilde\nu$ and can occur in the bistable region. The travelling front after the fold, develops a significant dip in its profile that eventually ends up as a multi-front and pulse solution terminating at the Maxwell point. For $\tilde\mu<0$, the fronts are known as `pushed'~\cite{saarloos2003} since the trivial state is unstable. In figure~\ref{f:perp_amp}(b) we have also plotted the linear spreading speed for the propagation of fronts into unstable states which determines the boundary between ``pushed" and ``pulled" fronts; see~\cite{saarloos2003}.

\section{Semi-analytical perturbation theory near the homoclinic snake}\label{s:semi_ana}
We review the semi-analytical perturbation theory near the homoclinic snake in 1D; see~\cite{burke2006,aranson2000}. Let $\delta$ be a small distance outside the homoclinic snaking region $\mu=\mu_{\mbox{fold}} + \delta$, with $|\delta|\ll1$. The transition time, $T$, between the passage of two successive nodes  is found by carrying out the perturbation
\begin{equation}\label{e:semi_ansatz}
u(x,t) = U_0(x) + |\delta|^{1/2}u_1(x,t),
\end{equation}
where $U_0(x)$ is a stationary localise pulse located at a fold on the homoclinic snake. Substituting~(\ref{e:semi_ansatz}) in to~(\ref{e:sh}) we find
\begin{equation}\label{e:semi_ana_1}
\mathcal{L}u_1 = \partial_tu_1 - |\delta|^{1/2}[\mbox{sgn}(\delta)U_0 + f_2u_1^2 + 3f_3U_0u_1^2 + 6f_4U_0^2u_1^2 + 10f_5U_0^3u_1^2] + \mathcal{O}(\delta),
\end{equation}
where $\mathcal{L}=\mathcal{L}[U_0]$ is the linearised SH operator at $\mu_{\mbox{fold}}$. This equation has to be solved subject to the requirement that $|u_1|\rightarrow0$ as $|x|\rightarrow\infty$. The righthand side of~(\ref{e:semi_ana_1}) is uniformly small provided $|\delta|\ll 1$ and hence the perturbation $u_1 $ evolves on a timescale $\mathcal{O}(|\delta|^{-1/2})$. At leading order, one has to solve the eigenvalue problem $\mathcal{L}u_1=0$. At a saddle node high up the snake, $\mathcal{L}$ has a three-dimensional kernel spanned by $\hat U_{e}(x)$ the even mode whose eigenvalue vanishes at the fold, $\hat U_{o}(x)$ the odd mode that tracks
the even mode ever more closely as one moves up the
snake, and $\hat U_n(x)$ the neutrally stable odd mode $U_0'(x)$.  Hence, we can write
\[
u_1(x,t) = a(t)\hat U_e(x) + b(t)\hat U_o(x) + c(t)\hat U_n(x) + \mathcal{O}(|\delta|^{1/2}),
\]
where $a,b,c$ are slowly evolving real amplitudes. Substituting this ansatz in to~(\ref{e:semi_ana_1}) and employing the Fredholm alternative on the righthand side of (\ref{e:semi_ana_1}) yields a system of ODEs for $a,b$ and $c$. The calculation is simplified by noting that in the space of reflection symmetric perturbations the ``centre of mass" of the pattern remains fixed. Hence, we can set $b\equiv c\equiv 0$ yielding 
\begin{equation}
\alpha_1\dot a = |\delta|^{1/2}[\alpha_2\mbox{sgn}(\delta) + \alpha_3a^2] + \mathcal{O}(|\delta|),
\end{equation}
where 
\begin{align}
\alpha_1 = \int_{-\infty}^{\infty}\hat U_e^2\D x,\qquad \alpha_2 = \int_{-\infty}^{\infty}U_0\hat U_e\D x,\qquad \alpha_3 = \int_{-\infty}^{\infty}(f_2  + 3f_3 U_0 + 6f_4U_0^2 + 10f_5U_0^3)\hat U_e^3\D x.
\end{align}
The quantities $\alpha_i$ need to be calculated numerically and hence why the calculation is semi-analytical. The transition time $T$ to pass between two successive folds is estimated as the time it takes $a$ to pass from $-\infty$ to $+\infty$ i.e.,
\begin{equation}\label{e:semi_ana_T}
T = \frac{\pi\alpha_1}{(|\delta|\alpha_2\alpha_3)^{1/2}} =: \alpha |\delta|^{-1/2}.
\end{equation}

\begin{figure}[h]
	\centering
	\includegraphics[width=0.8\linewidth]{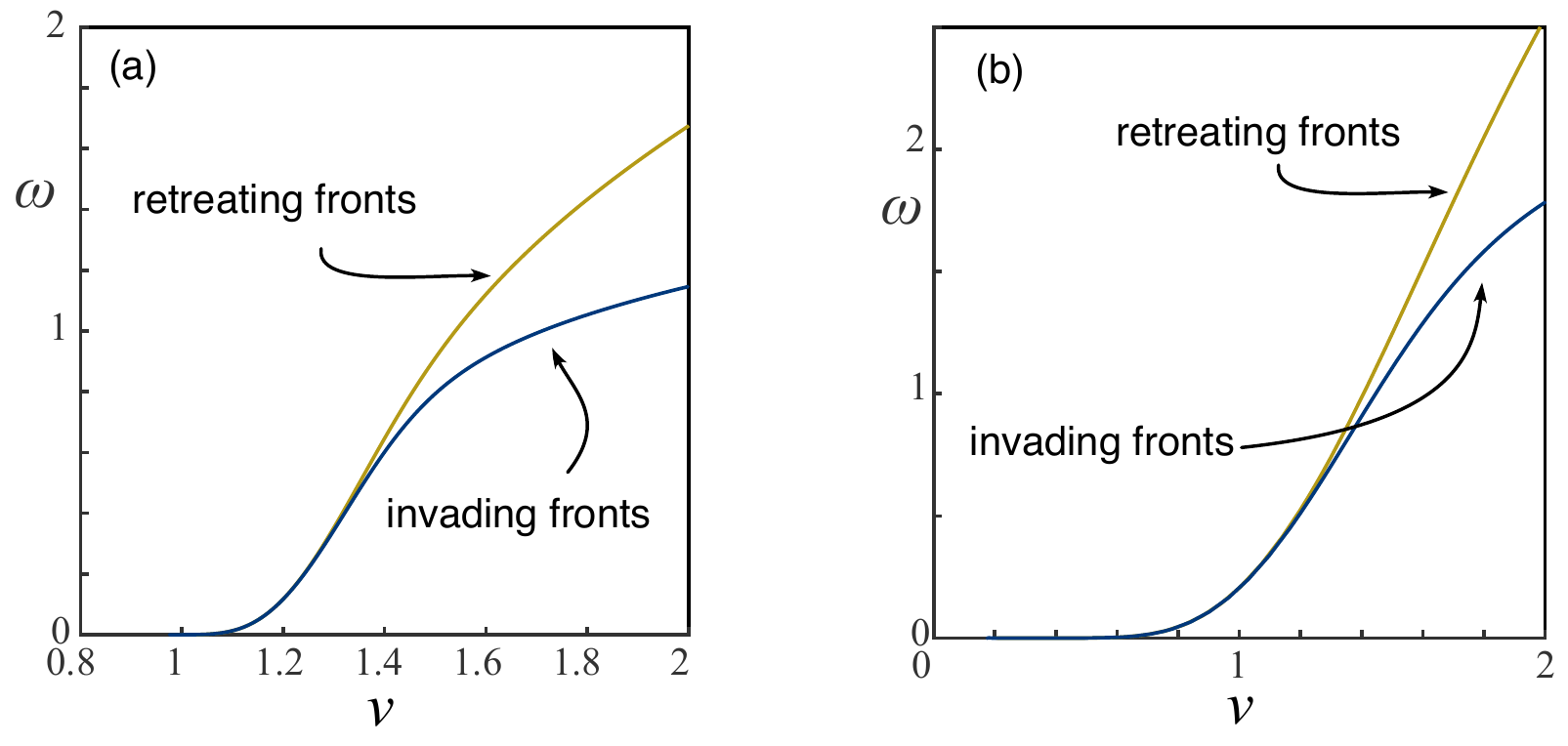}
	\caption{Semi-analytical prediction of the transition time pre-factor (a) quadratic-cubic (b) cubic-quintic SH equation. Gold line denotes trivial state invading while blue line denotes patterned state invading.\label{f:semi_ana_1D}}
\end{figure}

The integrals $\alpha_i$ vary strongly for localised pulses taken at the bottom of the snake (i.e. those with a small L2 norm) but start to converge as one takes pulses high up the snake. We therefore take a localised pulse sufficiently high up the snake so that the integrals $\alpha_i$ do not strongly depend on which localised pulse is taken. This corresponds to looking at invasion fronts.  Figure~\ref{f:semi_ana_1D} shows the transition frequency pre-factor $\omega = 2\pi/\alpha$ for both the quadratic-cubic and cubic-quintic SH equation in as $\nu$ is varied. We note that this frequency is approximately the invasion speed since $k_x\approx1$. We also plot the predicted transition frequency for retreating fronts close to the homoclinic snake where the trivial state invades the pattern. While these fronts do not select a unique transition frequency it is natural to ask what happens to a front whose stripes have the same wavenumber as that selected at the edge of the homoclinic snaking region. Here we see a shorter transition time $T$ compared to the invading fronts.

This calculation clearly has some problems as one would like to do a full series expansion of~(\ref{e:semi_ansatz}) but then this generates higher order terms that are not small in $\delta$ for the equation for $a$.  
Furthermore, this calculation tells us nothing about the selected spatial wavenumber of the pattern forming front only that it is close to that of the selected spatial wavenumber of the stationary localised pulse at the fold. 

\section{Numerical method for invasion fronts}\label{s:method}
In this section, we describe the numerical methods for time simulations, and boundary value problems for the perpendicular, oblique, parallel and almost-planar stripe invasion fronts. 

\subsection{Initial value solver}\label{s:ivp_method}
Time simulations of~(\ref{e:sh}) are carried out on a periodic rectangle using a 4th order exponential time-stepper Runge-Kutta scheme in time and fast Fourier method in space~\cite{kassam2005}. The scheme is implemented in \textsc{Matlab2017b} on a dual hexa-core (2.93 GHz) Mac Pro with 24GB RAM.

\subsection{Perpendicular stripes invasion fronts}\label{s:perpendicular_stripes_numerics}
We discuss the implementation of the continuation method for 2D perpendicular stripes. This computation is the easiest of those that we carry out since the asymptotic state (a perpendicular stripe) is constant in the front direction. Perpendicular travelling stripes solve the BVP
\begin{align}
cu_\rho - (1 + \partial_\rho^2 + k_y^2\partial_\gamma)^2u + f(u;\mu,\nu) =& 0,\qquad \rho\in[-L_\rho,L_\rho],\gamma\in(0,2\pi],\\
u(\rho,\gamma)-u(\rho,\gamma+2\pi)=&0,
\end{align}
where $u(x,y,t)=u(x-ct,k_yy)=:u(\rho,\gamma)$. We impose Neumann boundary conditions, $u_\rho=u_{\rho\rho\rho}=0$, at $\rho=\pm L_\rho$. To deal with the translational invariance symmetries in $\rho$ and $\gamma$, we impose the standard phase conditions~\cite{krauskopf2007}
\[
\int_{-L_\rho}^{L_\rho}\int_{-\pi}^{\pi}u^{\mbox{old}}_{\rho}(u-u^{\mbox{old}})\D \gamma\D \rho = 0,\qquad\mbox{and}\qquad \int_{-L_\rho}^{L_\rho}\int_{-\pi}^{\pi}u^{\mbox{old}}_{\gamma}(u-u^{\mbox{old}})\D \gamma\D \rho = 0.
\]
The resulting system is discretised with 4th order finite-differences in $\rho$ and pseudo-spectral Fourier collocation method in $\gamma$; see~\cite{trefethen2000}. The phase conditions are computed using the trapezoidal rule.  

\subsection{Oblique stripe invasion fronts}\label{s:oblique_stripe_numerics}
We now discuss the oblique stripe front continuation method. This method is an adaption of that described in Avery {\it et al.}~\cite{avery2018} in the context of oblique fronts with directional quenching.

We introduce the travelling coordinates $u(x,y,t) = u(x - c_xt,k_y(y-c_yt)) =:u(\rho,\gamma)$ and substituting into (\ref{e:sh}) we find
\begin{equation}\label{e:oblique_sh}
\mathbf{c}\cdot(u_\rho,k_yu_\gamma) -(1+\partial_\rho^2 + k_y^2\partial_{\gamma}^2)^2u - \mu u + f(u) =: Lu + f(u),
\end{equation}
where $\mathbf{c}=(c_x,c_y)$, and satisfies $\mathbf{c}\perp (k_x,k_y)$ the far-field wavenumbers i.e., $c_y= -c_xk_x/k_y$ so that the asymptotic patterns are stationary in the co-moving frame. The far-field stripes, $u_s(kx)=u_s(\xi)$ satisfy 
\[
-(1+k^2\partial_\xi^2)^2u_s - \mu u_s + f(u_s) = 0,\qquad u_s(0) = u_s(2\pi),
\]
where $k^2=k_x^2+k_y^2$. The BVP for stripe $u_s$ is solved using a pseudo-spectral Fourier collocation method with an additional phase condition
\[
\int_0^{2\pi}(\partial_\xi u_s^{\mbox{old}})(u_s-u_s^{\mbox{old}})\D\xi = 0,
\]
where $u_s^{\mbox{old}}=\cos(\xi)$, 
and then interpolated using a band-limited interpolant~\cite{trefethen2000}. 

For invading oblique stripes in the bistable region, it can be shown (see~\cite{avery2018}) that the linearisation about a front is Fredholm with index $-1$ in an exponentially weighted $L^2$-space and hence, we expect the fronts to come in one-parameter families parameterised by $k_y$ and select a front speed $c_x$ and far-field wavenumber $k_x$. Varying $k_y$, changes the angle of the stripes with respect to the front interface. 

We carry out a far-field core decomposition as for the fronts by setting~\cite{avery2018}
\begin{equation}
u(\rho,\gamma) = u_s(k_x\rho + \gamma;k)\chi(\rho) + w(\rho,\gamma),
\end{equation}
where $\chi(\rho)= (1+\mbox{tanh}(m(\rho-d)))/2$
and substituting this into ~\eqref{e:oblique_sh} to yield the inhomogeneous PDE for $w$ given by
\begin{equation}
L[u_s\chi + w] + f(u_s\chi + w) - \chi\left( Lu_h + f(u_s)\right) = 0,
\end{equation}
where we have subtracted to the far-field stripe solution. We add two phase conditions
\begin{align}
\int_0^{2\pi}\int_{L_x}^{L_x}([(\partial_\rho - k_x\partial_\gamma) ]w^{\mbox{old}})(w-w^{\mbox{old}})\D\rho \D\gamma = 0,\qquad  \int_0^{2\pi}\int_{L_\rho}^{L_\rho-2\pi/|k|}u_s'w\D\rho\D\gamma = 0.
\end{align}
Note, we do not need an additional phase condition to deal with the translational invariance in $\gamma$ since that symmetry has been fixed due to the choice of the far-field core decomposition (the translational invariance occurs as an asymptotic phase, $\theta$, in $u_s(k_x\rho+\gamma+\theta)$ which we set to zero). We again employ the same discretisation as used for perpendicular stripe fronts. The phase conditions are computed using the trapezoidal rule in both $\rho$ and $\gamma$. The Jacobian for the (now algebraic) system of equations, is explicitly computed with respect to the remainder function $w(\rho,\gamma)$ and a first order finite-difference is used for the Jacobian with respect to the asymptotic wave numbers $k_x$ and $c_x$. The resulting system is then solved using \textsc{matlab}'s \verb1fsolve1 routine. This is then embedded in a standard secant pseudo-arclength continuation routine as described in~\cite{krauskopf2007}.

\subsection{Parallel and almost-planar stripe invasion fronts}
In this section, we discuss how we implement the numerical boundary-value problem described in \S\ref{s:mod_fronts_inf} and the convergence of the method. We first describe our method for parallel stripe invasion fronts and then how we adapt it for almost-planar stripe invasion fronts.

Since parallel stripe invasion fronts have no $y$-dependence, we consider the 1D SH equation. The stationary parallel stripes of an invasion front in a co-moving frame are of the form $u_s(k_xx;k_x) = u_s(k_x(\xi+ ct);k_x),\xi,t\in\mathbb{R}$ and are time-periodic where $c>0$ is the propagation speed of the front. One-dimensional invasion fronts are of the form $u(x,t)=u(\rho,\tau)$, $\rho=k_x\xi,\tau=\omega t,\omega=ck_x$, and solve the boundary value problem 
\begin{subequations}\label{e:1D_BVP_numerics}
\begin{align}
\omega(u_{\rho} - u_{\tau}) - (1+k_x^2\partial_{\rho}^2)^2u - \mu u + f(u) =& 0,\qquad \rho\in[-L_\rho,L_\rho],\tau\in(0,2\pi],\\
u(\rho,\tau) - u(\rho,\tau+2\pi) =& 0,
\end{align}
\end{subequations}
We call $\omega>0$ the {\it transition frequency} and $k_x>0$ the {\it far-field wavenumber}.

Since the invasion fronts converge exponentially fast towards stripe patterns, we carry out a far-field core decomposition of the solution (see Appendix~\ref{s:modulated_finite})
\begin{equation}\label{e:farfield_core}
u(\rho,\tau) = u_p(\rho+\tau + \psi;k_x)\chi(\rho) + v(\rho,\tau; \omega),
\end{equation}
where we take $\chi(\rho)= (1+\mbox{tanh}(m(\rho-d)))/2$. Substituting this ansatz into the SH equation yields (after subtracting the equation for $u_p$)
\begin{equation}\label{e:PDE_w_2}
\mathbb{L}[u_p(\rho+\tau+\psi;k_x)\chi(z) + v(\rho,\tau;\omega)] + f(u_p(\rho+\tau+\psi;k_x)\chi(z) + v(\rho,\tau;\omega)) - \chi\left(\mathbb{L}u_p + f(u_p) \right) = 0,
\end{equation}
where $\mathbb{L}u :=  \omega(u_{\rho} - u_{\tau}) - (1+k_x^2\partial_{\rho}^2)^2u - \mu u$. Two additional integral phase conditions are added $\Phi_1v=\Phi_2v=0$ (see Appendix~\ref{s:modulated_finite}) where 
\begin{equation}\label{e:phase_cons_2}
\Phi_1 v =\int_0^{2\pi}\int_{-L_\rho}^{L_\rho}([\partial_\rho-\partial_\tau]v^{\mbox{old}})(v-v^{\mbox{old}})\D\rho\D\tau,\qquad \Phi_2v = \int_0^{2\pi}\int_{-L_\rho}^{-L_\rho+2\pi/|k_x|}u'_pv\D\rho\D\tau,
\end{equation}
and $v^{\mbox{old}}$ is a template solution (e.g. the initial guess or previous solution).

The boundary value problem for the 1D invasion fronts~(\ref{e:PDE_w_2}) is solved for $(v,\omega,k_x)$ on the computational domain shown in figure~\ref{f:comp_domain} which also shows the localised remainder function $v(\tau,\rho)$. We discretise $(\tau,\rho)$ using a pseudo-spectral Fourier collocation method in $\tau$ and fourth-order finite-differences in $\rho$. The 1D stripe $u_s$ is solved as for the oblique stripe front. The phase conditions~(\ref{e:phase_cons_2}) are computed using the trapezoidal rule in both $\rho$ and $\tau$. The Jacobian for the (now algebraic) system of equations, is explicitly computed with respect to the remainder function $v(\rho,\tau)$ and a first order finite-difference is used for the Jacobian with respect to the asymptotic wave numbers $k_x$ and $\omega$. The resulting system is then solved using \textsc{matlab}'s \verb1fsolve1 routine and embedded in a standard secant pseudo-arclength continuation routine as done for the oblique stripe fronts. 
\begin{figure}[h]
	\centering
	\includegraphics[width=0.9\linewidth]{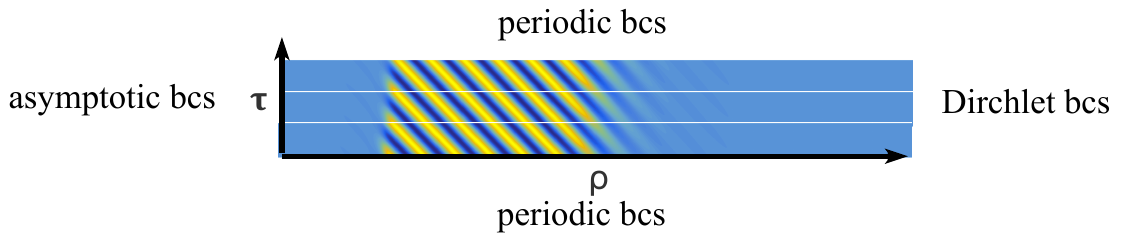}
	\caption{A plot of the remainder function $v$ and the computational domain for the parallel invasion fronts in $(\rho,\tau)-space$.  \label{f:comp_domain}}
\end{figure}

For almost planar fronts, we solve the boundary value problem
\begin{align*}
\omega(u_{\rho} - u_{\tau}) - (1+k_x^2\partial_{\rho}^2 + k_y\partial_\gamma^2)^2u - \mu u + f(u) =& 0,\qquad \rho\in[-L_\rho,L_\rho],\tau\in(0,2\pi],\gamma\in(0,2\pi],\\
u(\rho,\tau,\gamma) - u(\rho,\tau+2\pi,\gamma) =& 0,\\
u(\rho,\tau,\gamma) - u(\rho,\tau,\gamma+2\pi) =& 0,
\end{align*}
where $u(\rho,\gamma,\tau)=u(k_x(x- ct),k_yy,\omega t)$, by carrying out a far-field core decomposition (\ref{e:farfield_core}) as before. To deal with the translational invariance in $\gamma$ of solutions, we impose an additional phase condition
\[
\Phi_3 v = \int_0^{2\pi}\int_0^{2\pi}\int_{-L_\rho}^{L_\rho}(\partial_\gamma v^{\mbox{old}})(v-v^{\mbox{old}})\D\rho\D\tau\D\gamma,
\]
and the other phase conditions (\ref{e:phase_cons_2}) are now integrated also over $\gamma\in(0,2\pi]$.

We discretise in $\gamma$ using a pseudo-spectral Fourier collocation method as in the $\tau$ direction and discretise $(\rho,\tau)$ as before.

We have found that  \textsc{matlab}'s \verb1fsolve1 routine works reasonably well though it can become slow for large discretisations where a Newton solver that takes advantage of the `bordered' structure of the Jacobian matrix can greatly speed up the computations. The main advantage for using \verb1fsolve1 is that it is able to converge despite the initial guess having a large residual. 

\begin{figure}[p]
	\centering
	\includegraphics[width=\linewidth]{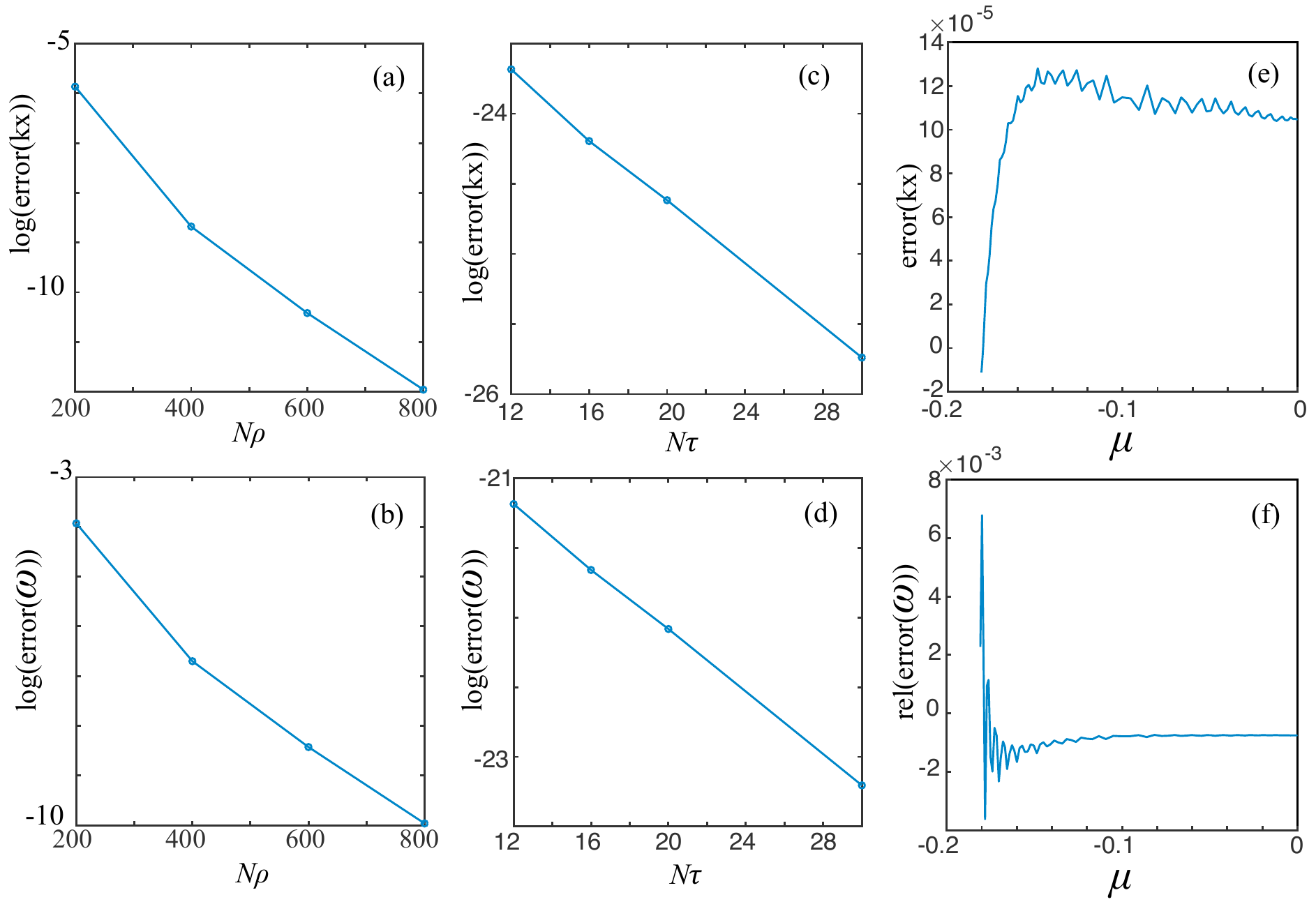}
	\caption{Absolute error from a high resolution solution at $(\mu,\nu)=(0,1.6)$, (a) \& (b) $N_\tau=20,L_\rho=40\pi,d=40$, (c) \& (d) $N_\rho=400,L_\rho=40\pi,d=40$, (e) \& (f) $N_\tau=20,N_\rho=400,L_\rho=40\pi$\label{f:convergence_1}}
\end{figure}

\begin{figure}[p]
\centering
\includegraphics[width=0.6\linewidth]{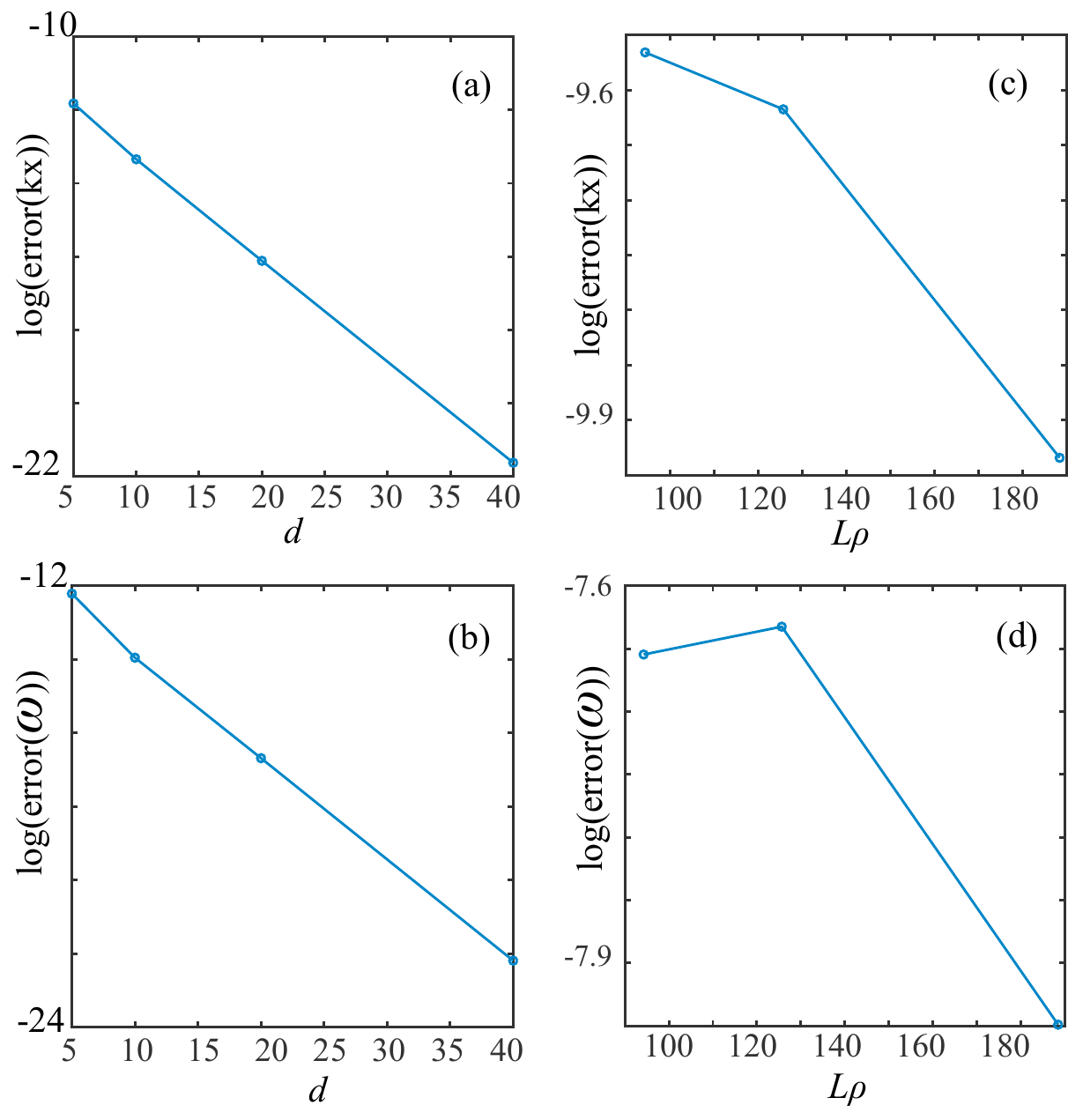}
\caption{Absolute error (a) \& (b) $N_\tau=20,N_\rho=400,L_\rho=40\pi$, (c) \& (d) $N_\tau=20,N_\rho=800,d=20$.\label{f:convergence_2}}
\end{figure}

Initial conditions are chosen for the fronts to be
\begin{align*}
u_{\mbox{perpendicular}}(\rho,\gamma) =&\frac12(1-\mbox{tanh}(\rho))\cos(\gamma),\\
u_{\mbox{oblique stripes}}(\rho,\gamma) =& \frac12(\mbox{tanh}(\rho+40)-\mbox{tanh}(\rho))u_s(\rho+\gamma;0.99),\qquad k_y=0.75,\\
u_{\mbox{parallel}}(\rho,\tau) =& \frac12(\mbox{tanh}(\rho+40)-\mbox{tanh}(\rho))u_s(\rho+\tau;0.99),\qquad \omega=0.99,\\
u_{\mbox{almost planar}}(\rho,\gamma,\tau) =& u_{\mbox{parallel}}(\rho,\tau) +  \frac12\left(\mbox{tanh}\left(\rho+\frac{\pi}{2}\right)-\mbox{tanh}\left(\rho-\frac{\pi}{2}\right)\right)\cos(\gamma)\cos(\rho).
\end{align*}

Typical computational meshes are $N_\rho=400,L_\rho=40\pi,N_\tau=20$ for 1D invasion fronts and oblique stripes. For 2D almost planar invasion fronts, we use $N_\tau=N_\gamma=16$.  Refinements of these meshes do not significantly change the results. 

Plots showing the convergence of the method for parallel stripe fronts are presented in figures~\ref{f:convergence_1} and \ref{f:convergence_2} where we take $(\mu,\nu)=(0,1.6)$ for the quadratic-cubic SH equation. Here we see that with relatively few collocation points in $\tau$ one finds rapid convergence. For the discretisation in $\rho$, we see as expected algebraic convergence due to the finite-differences. We also expect that as one approaches the edge of the homoclinic snaking, the invasion speed slows down and the numerical approximation should break down. In figure~\ref{f:convergence_1}, we plot the relative error of $\omega$ (since $\omega$ is going to zero absolute error would not be a good measure) and we see that the error remains small until one is very close to the homoclinic snake. 

In figure~\ref{f:convergence_2}, we explore the dependence on the computational parameters $d$ (the cut-off location) and $L_\rho$ (the finite truncation of $\rho$). 
We next show rapid convergence of the selected far-field wavenumber and invasion speed as we varying the cutoff parameter $d$ and the domain truncation parameter $L_\rho$ for 1D invasion fronts in the quadratic-cubic SH equation with $\nu=1.6$. We define the error in the selected far-field wavenumber and invasion speed to be the difference from the respective values but for $N_{\rho}=1000,N_{\tau}=40,L_\rho=40\pi,d=80$. We again see good convergence of the selected parameters and that they are insensitive to $d$ and $L_\rho$.

We have also verified the wavenumber selection and transition frequency seen in figure~\ref{f:invade_qc_1D} using the initial value solver in \S\ref{s:ivp_method}.

\section{Results}\label{s:numerics}

\subsection{Parallel invasion fronts}\label{s:rolls}
We start by looking at parallel invasion fronts in the quadratic-cubic SH equation for $\nu=1.6$. For these parameters the Maxwell point occurs at $\mu=0.2$ and snaking occurs for $0.181\leq \mu\leq 0.211$. In figure~\ref{f:invade_qc_1D}, we plot the invasion fronts where the stripes invade the trivial state. We see that the selected far-field wavenumber $k_x$ starts at the edge of the homoclinic snaking at the same value as the stripes with zero Hamiltonian (\ref{e:1D_Ham}). The selected far-field wavenumber then decreases to a minimum value around $\mu\sim0.1$ before increasing. Numerically, we also find invasion fronts for $\mu<0$, outside the bistability region until $\mu=\mu_c$ where the trivial state loses convective stability. It is interesting to note that the selected far-field wavenumber's behaviour is qualitatively very different to that predicted from the weakly nonlinear analysis in \S\ref{s:weak} where it is predicted to monotonically decreases as one approaches the edge of the homoclinic snaking region. To explore this a bit further, we trace out in figure~\ref{f:invade_qc_1D}(c) the selected far-field wavenumber $k_x$ in $(\mu,\nu)$-space and trace out the location of the minimum value of $k_x$. Here we see that qualitatively this ``dip" in $k_x$ occurs for all values in $(\mu,\nu)$-space but the location of the minimum value of $k_x$ approaches the edge of the homoclinic snaking. This suggests that the exponentially small asymptotics of Kozyreff \& Chapman~\cite{kozyreff2006,kozyreff2009} should be able to capture this behaviour of $k_x$.

\begin{figure}[h]
\centering
\includegraphics[width=0.8\linewidth]{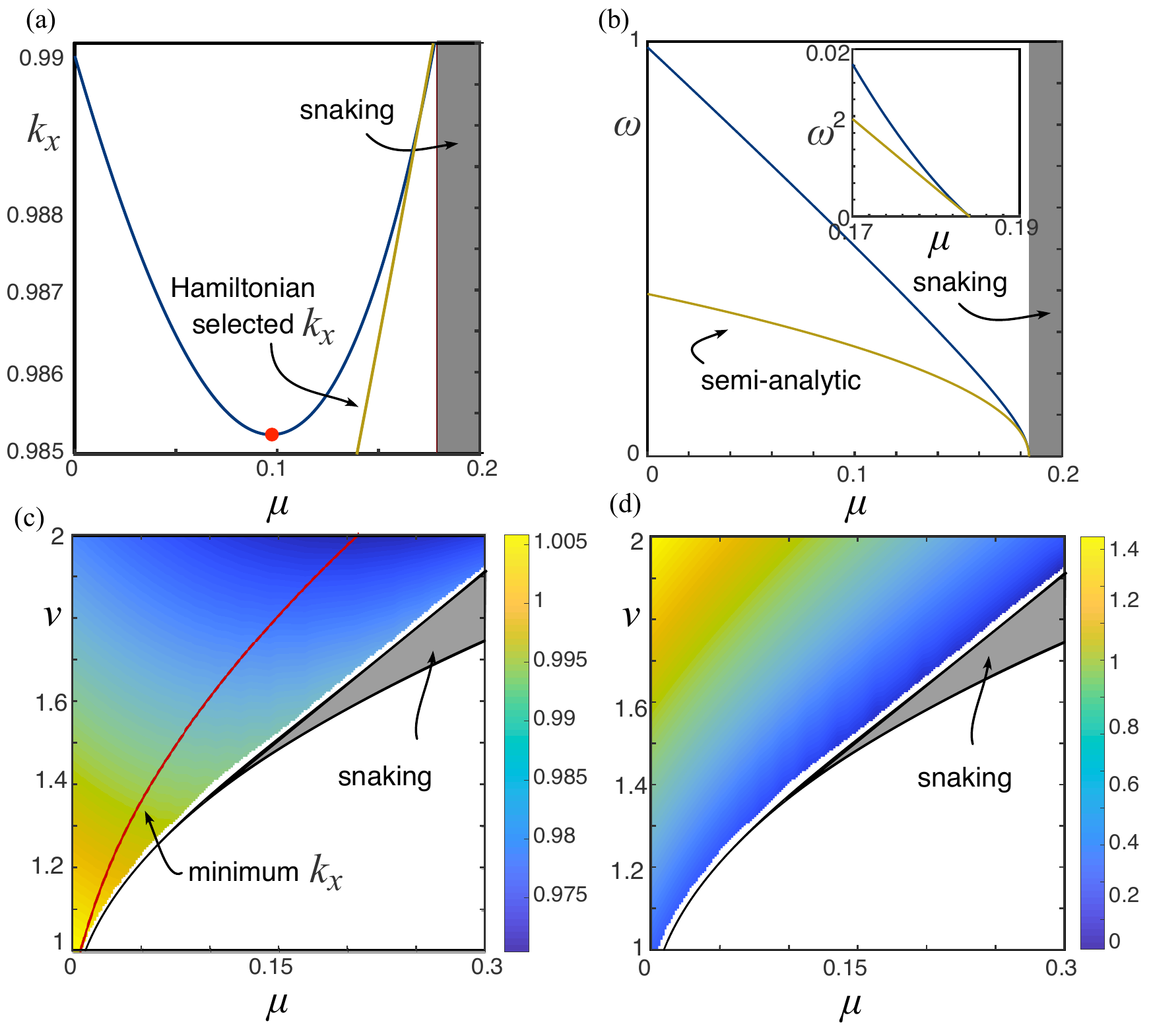}
\caption{(a) Bifurcation diagram showing the pattern selection wave number, $k_x$, for the parallel invasion fronts in the quadratic-cubic SH equation with $\nu=1.6$. The edge of the snaking region occurs around $\mu=0.181$ with a Hamiltonian selected wavenumber $k_x\sim0.9905$. The gold solid curve is the stationary stripes selected by the Hamiltonian condition. (b) shows the the invasion temporal wavenumber $\omega$ and in gold the semi-analytic prediction. The inset panel plots $\omega^2$ against $\mu$ to compare the semi-analytic prediction. (c) and (d) show the two parameter sweep of the selected wavenumber $k_x$ and invasion wavenumber $\omega$, respectively. The snaking region is shown as a grey shaded region and location of the minimum value of $k_x$ is plotted as a red curve in (c). \label{f:invade_qc_1D}}
\end{figure}

We also plot the selected transition frequency, $\omega$, in figure~\ref{f:invade_qc_1D}(b). We note that since $k_x\approx1$, the wave speed of the front is approximately $\omega$. Here we see the $|\delta|^{1/2}$-scaling predicted by the semi-analytical perturbation theory in \S\ref{s:semi_ana}, where $\delta$ is the distance from the edge of the homoclinic snaking region. Qualitatively, we see the behaviour of $\omega$ in $(\mu,\nu)$-space remains the same.

\begin{figure}[h]
	\centering
	\includegraphics[width=0.4\linewidth]{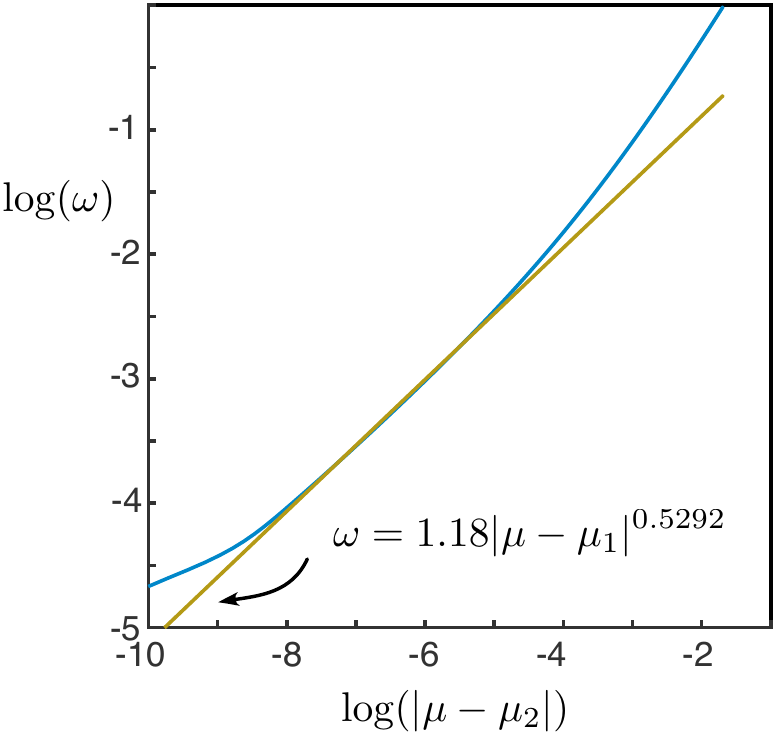}
	\caption{A plot of the log of the transition frequency for invasion fronts to the right of the snaking region $\nu=1.6$. Yellow line are least square fit. The semi-analytical prediction is $\omega=1.29|\mu-\mu|^{0.5}$.\label{f:semi_ana_comp}}
\end{figure}

In order to do a comparison with the semi-analytical perturbation theory, in figure~\ref{f:semi_ana_comp} we plot the log of the selected transition frequency. We observe for invading fronts where the trivial state invades, for sufficiently large $\nu$, that there is indeed a square root scaling law with respect to the distance from the edge of the homoclinic snake. For small values of $\nu$, we do not find a good fit which we speculate is due to the presence of the fold of the stripes. The least squares fit with the data compares reasonably well with the semi-analytical prediction but only very close to the edge of the homoclinic snaking region. However, it should be highlighted that numerics become increasingly difficult as one approaches the edge of the homoclinic snake as $\omega$ tends to zero and a finer discretisation is required.

\begin{figure}[h]
	\centering
	\includegraphics[width=0.8\linewidth]{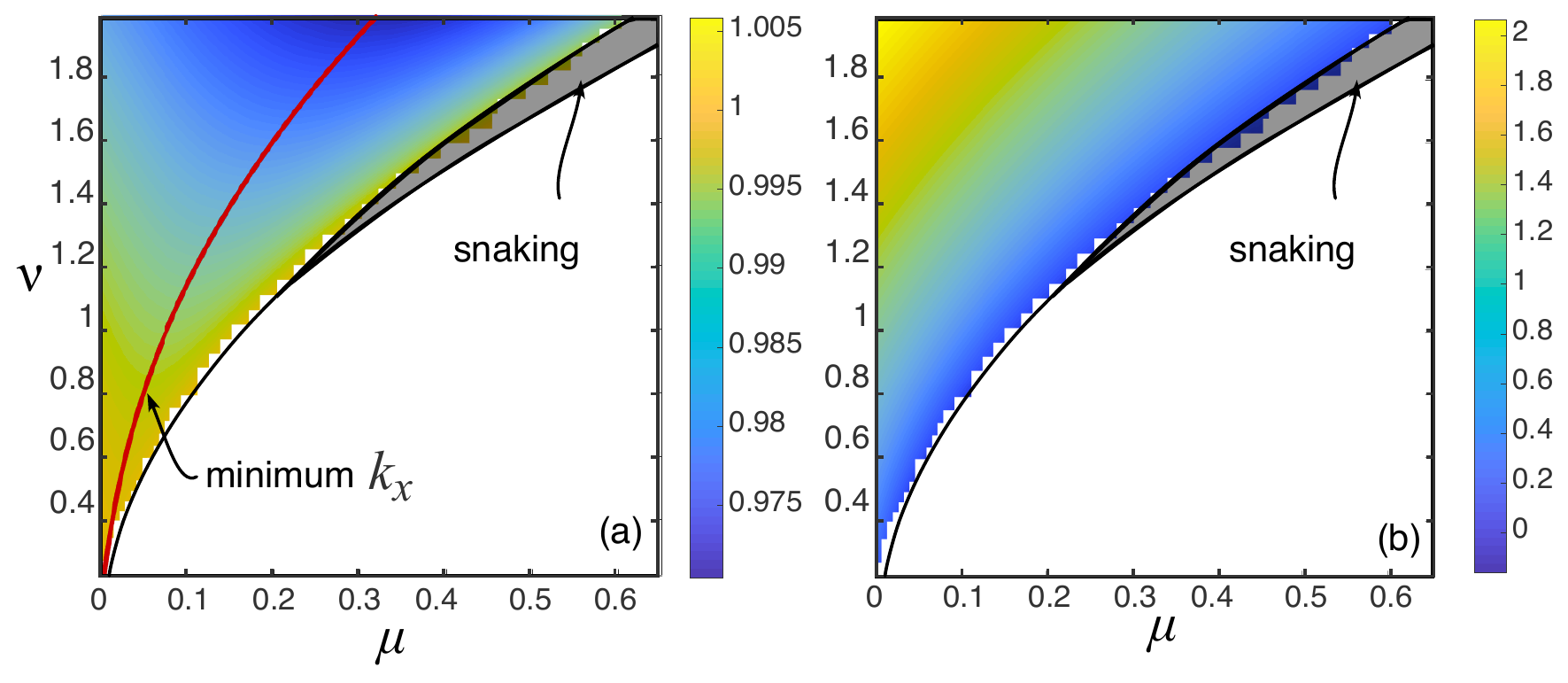}
	\caption{Two-parameter bifurcation diagrams for parallel invasion fronts in the Cubic-Quintic SH equation (a) $k_x$ and (b) $\omega$. The snaking region is shown as a shaded grey area and the location of the minimum selected $k_x$ wavenumber is plotted as a red curve in (a). \label{f:1D_fronts_cq}}
\end{figure}

We conclude our investigation by plotting the selected $k_x$ and $\omega$ values for the cubic-quintic SH equation in figure~\ref{f:1D_fronts_cq}. Here we see qualitatively similar behaviour as for the quadratic-cubic SH equation.  

\subsection{Perpendicular \& oblique planar stripes}\label{s:olbique_stripes}

In this section, we look at invasion fronts on the plane involving perpendicular and oblique stripes. We concentrate only on the cubic-quintic SH equation as stripes are typically destablised by hexagonal perturbations in the quadratic-cubic SH equation.

In figure~\ref{f:perp}, we plot the selected front speed for various $k_y$ values for $\nu=1.25$. We find good qualitative agreement with the weakly nonlinear analysis in \S\ref{s:weak} where the fronts form ``loops" in $(\mu,c)$-space. For $\tilde\mu<0$, the fronts are known as `pushed'~\cite{saarloos2003} since the trivial state is unstable and so we also plot the linear spreading speed~\cite{saarloos2003} for $\mu<0$ and $k_y=1,1.025,1.05$. We note that perpendicular localised stripes do not undergo homoclinic snaking since in a spatial dynamics formulation the stripes are constant in $x$; see~\cite{avitabile2010}. The upper branches are stable with respect to co-periodic perturbations. In order to understand their stability on the plane, we trace out the loci corresponding to the folds in figure~\ref{f:perp}(a) and the zero wavespeed as we vary $k_y$ in figure~\ref{f:perp}(b). We find a large region where the stripes in the invasion fronts are linearly stable in the plane. For these parameter values the zig-zag instability occurs at $k_y\sim1$ (in fact a little less than 1). We see that the stripes of the invasion front can be destabilised by either an Eckhaus or zig-zag instability. 
\begin{figure}[h]
	\centering
	\includegraphics[width=0.8\linewidth]{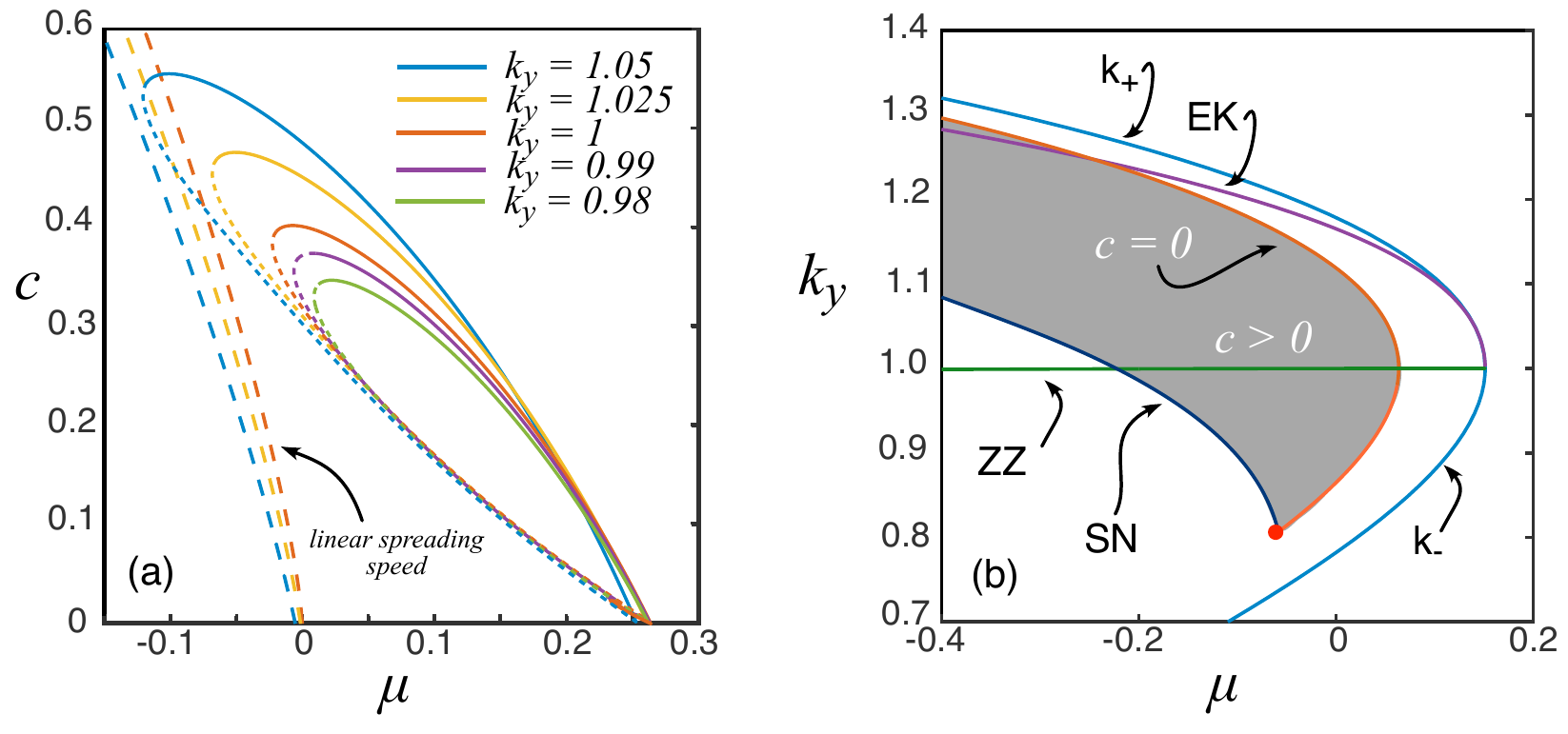}
	\caption{Perpendicular stripe fronts in the cubic-quintic SH equation with $\nu=1.25$, (a) existence for various $k_y$'s (b) two parameter bifurcation diagram depicting the region of existence of invading perpendicular fronts. The bifurcation point where the saddle node and $c=0$ curves intersect is at $(\mu,k_y)=(-0.141,0.807)$. $k_{\pm}$ denotes the existence boundaries for the stripes,  ``EK" denotes the Eckhaus instability boundary, and ``ZZ" the zig-zag instability boundary.\label{f:perp}}
\end{figure}

In figure~\ref{f:oblique}, we plot the selected far-field wavenumber $k = \sqrt{k_x^2+k_y^2}$ for various $k_y$ values and $\nu=1.25$ for oblique invasion fronts. For $k_y=0$, we have the parallel stripes described in the previous section. We observe the characteristic ``dip" in the selected far-field wavenumber which increases as the stripes become more slanted. The existence curves can be continued up to the Maxwell point. Interestingly, we can only find oblique stripe fronts for $k_y<0.8$ and there appears to be a critical angle for the oblique stripes beyond which there do not exist. All these oblique stripes invade slower in the $x$-direction than the parallel stripes except very close to or near the 1D snaking region where the reverse is true; see figure~\ref{f:oblique}(c).  Far away from the homoclinic snaking region, this is consistent with the weakly nonlinear analysis of oblique stripes described in \S\ref{s:weak} except where pinning effects occur near the homoclinic snaking region. However, the overall invasion speed $|\mathbf{c}|$ is larger than that for the parallel stripes. We note that the results in  Goh \& Scheel~\cite{goh2017} suggest that the fastest front should be the one with largest $k_y$ near the homoclinic snaking region and that more refined numerical experiments are possibly required to resolve this region.

\begin{figure}[h]
	\centering
	\includegraphics[width=\linewidth]{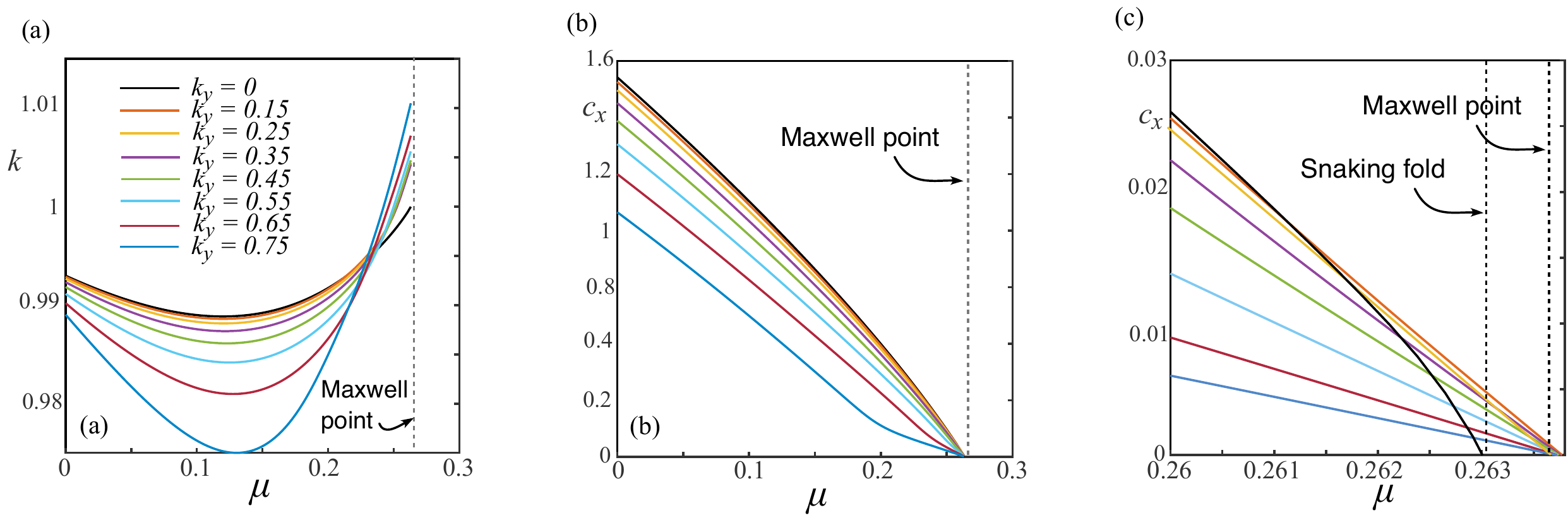}
	\caption{Oblique stripe invasion fronts in the cubic-quintic SH equation with $\nu=1.25$ with ($k_y>0$) and parallel ($k_y=0$) stripe fronts (a) the selected wavenumber, $k$, (b) wavespeed, $c_x$, selection in the $x$-direction, and (c) a zoom-in near the Maxwell point.\label{f:oblique}}
\end{figure}

\subsection{Almost planar stripe invasion fronts}\label{s:almost_stripes}

\begin{figure}[p]
\centering
\includegraphics[width=0.8\linewidth]{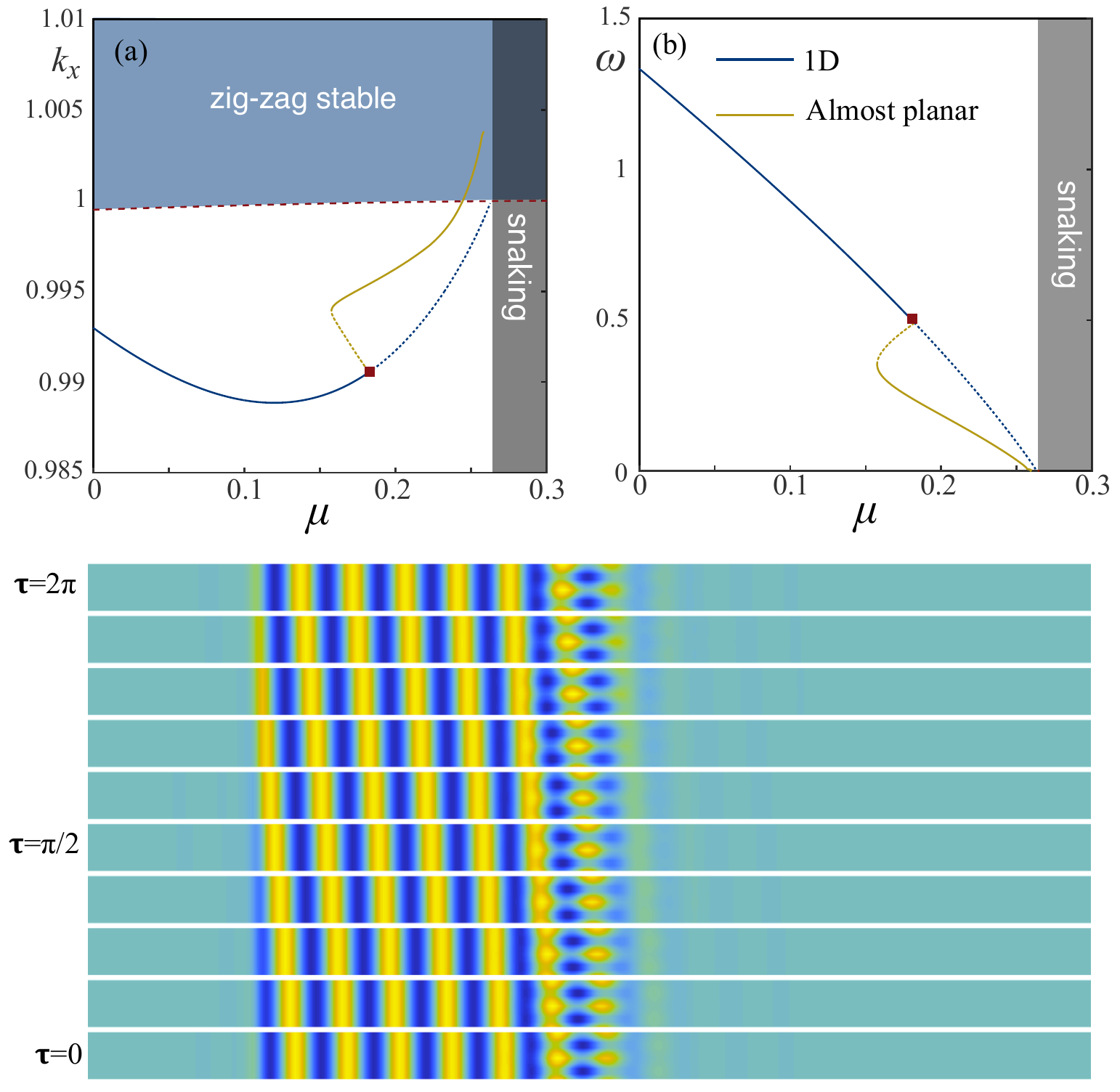}
\caption{Bifurcation diagrams of the almost planar invasion fronts $\nu=1.25,k_y=0.5$, ($N_\gamma=N_t=16,N_\rho=400,L_\rho= 40\pi$) and the parallel invasion fronts. Stability is respect to perturbations with $k_y=0.5$ but IVP suggests stability for smaller $k_y$ values. Panel (a) shows the selected far-field stripe wavenumber while panel (b) shows the selected temporal invasion wavenumber. The parallel invasion fronts are plotted as a blue curve while the almost planar invasion fronts are plotted as a gold curve. Panel (c) shows the ``core" of an almost planar invasion front at different time slices for $\mu=0.2,\nu=1.25,k_y=0.5$ and $k_x=0.9962, \omega =0.1749$\label{f:almost}}
\end{figure}
\begin{figure}[p]
	\centering
	\includegraphics[width=0.4\linewidth]{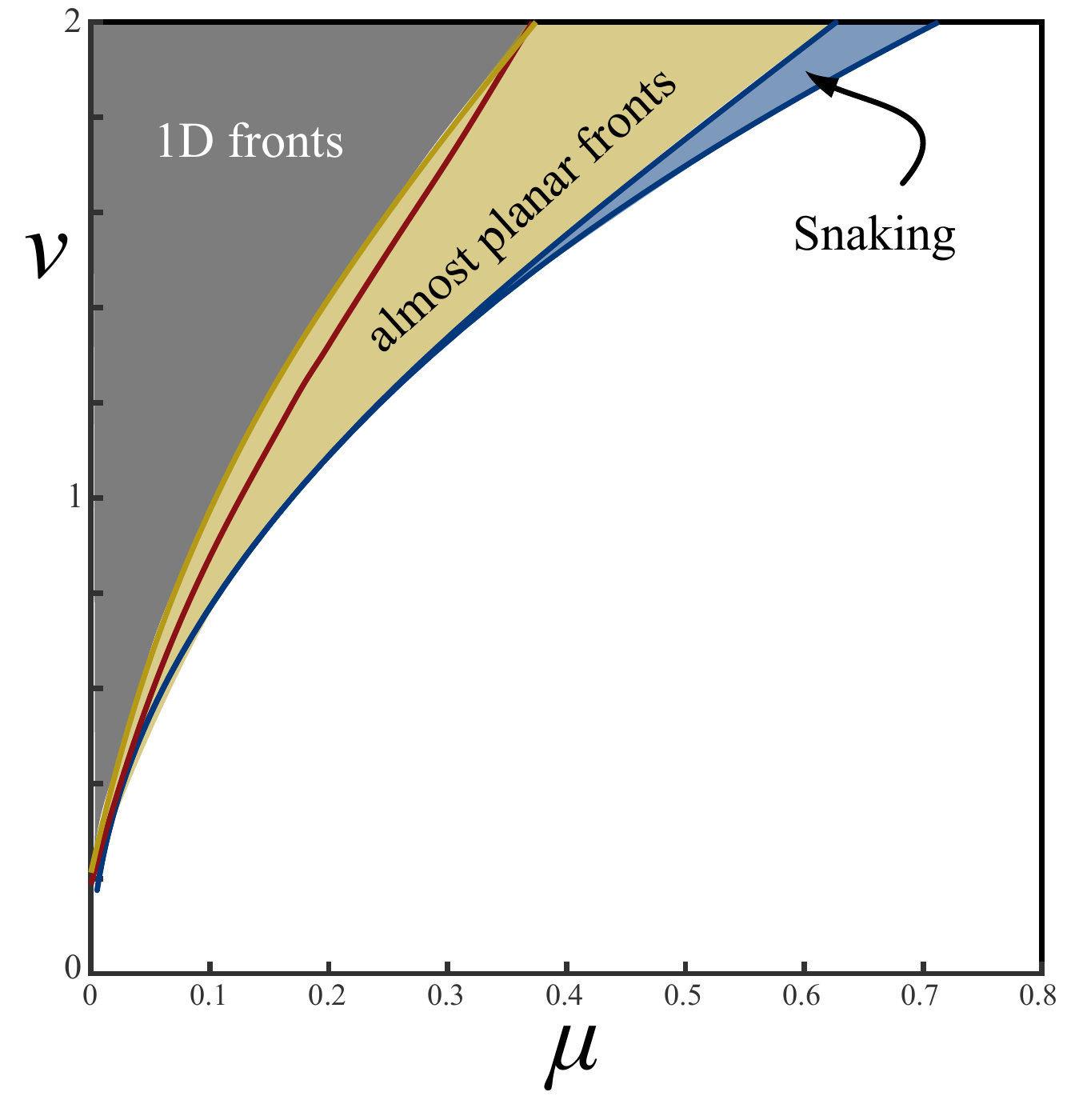}
	\caption{Two parameter bifurcation diagram for cubic-quintic SH equation. The red curve denotes the bifurcation from the parallel invasion front while the gold curve denotes the fold of the almost planar fronts. \label{f:2params_almost}}
\end{figure}

We now look at almost planar invasion fronts in the cubic-quintic SH equation. In figure~\ref{f:almost}, we show the bifurcation diagram for the almost planar and parallel invasion fronts with $\nu=1.25$ with $k_y=0.5$ fixed. Here we see that the selected far-field stripes of the parallel invasion fronts are zig-zag unstable for the entirety of the bistable region (though do restabilise for sufficiently large negative $\mu$, not shown). However, for moderate values of $k_y$, the zig-zag instability is not seen since it is a long wavelength instability. Interestingly, we observe a finite transverse bifurcation off the parallel invasion fronts at $\mu\sim0.18$ which leads to the formation of an almost planar invasion stripe front. One of these almost planar stripe invasion fronts is shown in figure \ref{f:almost}(c) at $\mu=0.2$. We find the bifurcation is subcritical and then restabilises at a fold at $\mu=0.16$ where it becomes stable with respect to co-periodic perturbations. The speed of these almost planar stripe invasion fronts is always found to be slower than the parallel invasion fronts. As we approach the snaking region, this almost planar stripe invasion front selects a far-field stripe wavenumber that is zig-zag stable and hence we expect to be able to see these fronts on the plane. We note that the almost planar fronts cannot approach the edge of the 1D snaking region since stationary almost planar fronts snake in a far larger region in parameter space.

In figure~\ref{f:almost_ivp}, we see that depending on if one starts either side of the fold of the almost planar stripe fronts then one sees the growth of stable almost planar fronts or the initial transverse instability dies out. We have been able to take large $y$ domains and observe this behaviour and no zig-zag instability. 

We are able to trace out in $(\mu,\nu)$-space the locations of the bifurcation loci corresponding to the snaking region, the bifurcation to the almost planar stripe fronts, and the fold of the almost planar fronts; shown in figure~\ref{f:2params_almost}. The locus of the fold of the almost planar fronts denotes the critical $(\mu,\nu)$-values for beyond which co-existence between parallel invasion fronts is observed. We observe that all these loci persist to small values of $\nu$ suggesting a weakly nonlinear analysis near the co-dimension 2 point $(\mu,\nu)=(0,0)$ might be able to capture this bifurcation to almost planar invasion fronts; see \S\ref{s:discussion} for a discussion on this.

\section{Planar patch invasion}\label{s:patch_invasion}
In this section, we will investigate what happens to patches of stripe pattern on the plane (known as ``worm patches" \cite{avitabile2010}) that invade the trivial state in the cubic-quintic SH equation.
\begin{figure}[h]
	\centering
	\includegraphics[width=0.8\linewidth]{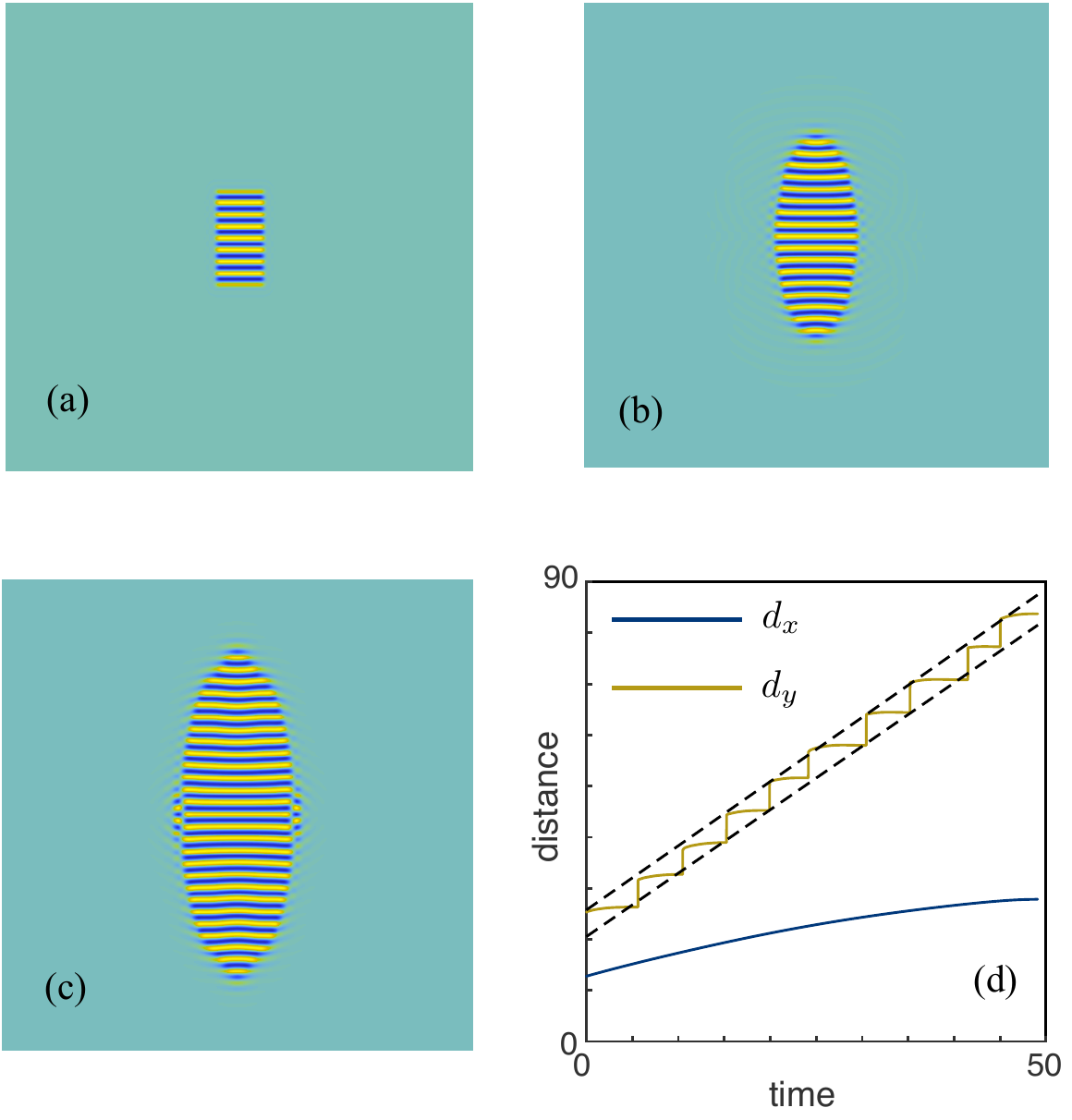}
	\caption{Worm patch invasion in the cubic-quintic SH equation with $(\mu,\nu)=(0.01,1.25)$. (a) at $t=0$, (b) $t=25$, (c) $t=50$ (d) the interface locations in the $x$- and $y$-directions at the midpoints given by $d_x$ and $d_y$, respectively. The upper dashed black line has a fitted line of $1.26t+25.75$ and the lower line of $1.24t+20.53$.\label{f:worm_patch_1}}
\end{figure}

We  carry out time simulations and choose a computational domain of $(x,y)\in[-60\pi,60\pi]^2$ with $N_x=N_y=2^{10}$ points in $x$ and $y$. 
The simulations are initialised with a rectangular patch of stripes given by
\[
u(x,y,0) = \frac{1.2}{4}(\mbox{tanh}(x+d_1) - \mbox{tanh}(x-d_1))(\mbox{tanh}(y+d_2)-\mbox{tanh}(y-d_2))\cos(y),
\]
where $(d_1,d_2)=(4\pi,8\pi)$. In figure~\ref{f:worm_patch_1}, we take $(\mu,\nu)=(0.01,1.25)$ and show the plots of the patch at $t=0,25,50$. We detect the location of the $x$ and $y$ interfaces (denoted by $d_x$ and $d_y$, respectively) along the mid-points by looking for the last zeros of $\tilde u-0.5$ using \textsc{Matlab}'s \verb1fzero1 command where $\tilde u$ is a piecewise linear interpolant of the discretised $u$ found using \textsc{Matlab}'s \verb2interp12 command. 

After the initial transient, we see in figure~\ref{f:worm_patch_1} the patch invades in both the $x$-direction via a perpendicular stripe front and in the $y$-direction via parallel invasion front. Plotting the location of the $y$-direction interface, $d_y$, we see that it undergoes a stepping process as expected with an average speed of $1.2$. The mean time between jumps is $4.96$ which is close to that predicted from the parallel invasion fronts $4.8$. At these values, there is no stable perpendicular invasion front for the wavenumber selected by the parallel invasion front. Since the propagation of the parallel invasion fronts is always quicker than the perpendicular front, the patch starts to bulge along the $x$-direction while in the $y$-direction the patch grows to a point. We see that the oblique stripes develop curvature so that their interface becomes perpendicular. All this results in a lengthening of the wavelength of the stripes in the middle of patch. From figure~\ref{f:perp}, we see that as the wavelength increases (corresponding to decreasing wavenumber) the perpendicular invasion stripes slow down and cease to exist.

\begin{figure}[h]
	\centering
	\includegraphics[width=\linewidth]{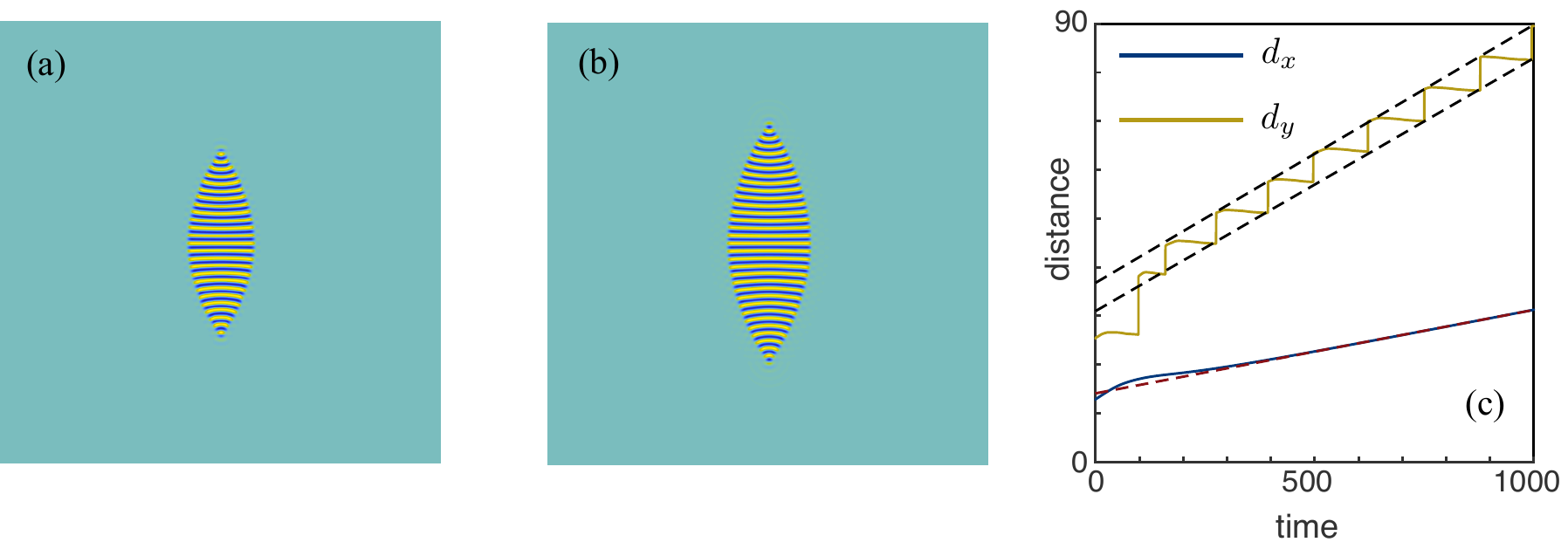}
	\caption{Worm patch invasion for the cubic-quintic SH equation with $(\mu,\nu)=(0.24,1.25)$. (a) $t=500$, (b) $t=1000$, $(c)$ the interface locations in the $x$- and $y$-directions at the mid-points given by $d_x$ and $d_y$, respectively. The upper dashed black line has a fitted line of $0.5t+36.8$ and the lower line of $0.5t+31.0$. The dashed red line has a fitted line of $0.02t+14.1$. \label{f:worm_patch_2}}
\end{figure}

In figure~\ref{f:worm_patch_2}, we increase $\mu$ to $0.24$ and keep $\nu=1.25$ fixed. For these parameter values  we see that the invasion speed of the perpendicular front in the $x$-direction settles down to a constant suggesting that the wavenumber of the stripes in the perpendicular front has been selected. We note that the mean time between jumps is about $112$ which is close to that for the almost planar invasion front $\sim90$ (for the parallel invasion front the mean jump is $39$). It appears this growing patch is stable and will become a long (in the $y$-direction) thin (in the $x$-direction) patch. The selected stripes are expected to be zig-zag unstable in this region. 

We note that close to the homoclinic snaking region, figure~\ref{f:oblique}(c) would suggest that the patches may change shape as the perpendicular front will travel faster than the parallel front leading to wide (in the $y$-direction), short (in the $x$-direction) patches. However, we have not been able to find this region and we always observe that the parallel front direction of the patch always invades faster than the perpendicular front direction. 

\section{Conclusion}\label{s:discussion}

{\bf Summary.} In this paper, we have provided a pattern selection principle for invasion fronts involving stripes in their far-field outside the homoclinic snaking region and, implemented a novel numerical continuation routine to systematically explore the invasion fronts in the bistable region of the planar SH equation. We have found that  the parallel invasion fronts propagate faster as they move away from the homoclinic snaking region and have a dip in the selected far-field wavenumber. 
In 2D, we explored a range of stripe fronts (perpendicular and oblique) as well as almost planar stripes fronts. We found that the bifurcation picture of the invasion fronts is far richer in 2D in the cubic-quintic SH equation.  
Finally, we showed how the front computations help explain the invasion process of patches of stripe pattern on the plane in the cubic-quintic SH equation. 
Now that we have a numerical method to compute stripe invasion fronts, we can now explore several open problems which we list below. 

{\bf Open Problems.} 

In terms of stationary fronts, it appears that oblique stripe fronts have been overlooked. It is clear that due to the Hamiltonian structure, oblique fronts cannot snake and can only exist at the 1D stripe Maxwell point. However, this raises the question of can they snake in systems which do not have spatial conserved quantities? In general, we do not expect snaking of oblique stripe fronts since the intersection between the centre-unstable and stable manifolds for such fronts in a spatial dynamics setting is not transverse due to the two-dimensional kernel from the $x$ and $y$ derivative.  In terms of planar stripe patch invasion, oblique fronts appear to have very little to do with the growth of the patch. In particular, it appears the patch can be well explained by the parallel and perpendicular invasion fronts only. 

There appears to be a significant qualitative difference between the predicted wavenumber selection of the parallel invasion fronts from the weakly nonlinear amplitude equation analysis and the numerics. In particular, the monotonic decrease in the far-field wavenumber as one approaches the snaking region from the amplitude equations is not observed numerically. Numerically, it appears the dip in the wavenumber converges to the edge of the homoclinic snaking region as one approaches the co-dimension 2 bifurcation point where the weakly nonlinear analysis becomes applicable. It would be interesting to look at the exponential asymptotics of Kozyreff \& Chapman~\cite{kozyreff2006,kozyreff2009,kozyreff2013} to see if the qualitative dip of the selected wavenumber can be explained analytically. 

It would be interesting to see if the weakly nonlinear analysis can explain the bifurcation from the parallel invasion fronts to the almost planar invasion fronts.  Close to the co-dimension 2 point, $(\mu,\nu)=(0,0)$, in the cubic-quintic SH equation one could carry out a weakly nonlinear analysis to derive amplitude equation and then investigate the bifurcation of the parallel invasion fronts to almost planar stripe fronts.

There are some interesting 2D instabilities of the invasion fronts that could be further investigated. It may be possible in a different system, that the selected stripe for the parallel invasion front becomes zig-zag stable leading to a bifurcation in the essential spectrum. We were unable to find a period-doubling (in time) bifurcation but this may be possible as the fronts get quicker. 

 To deal with the retreating fronts and invasion fronts into unstable states, one would have to impose an additional constraint to make solutions of the boundary value problem unique again. Coming up with a more dynamical systems motivated constraint would be interesting to explore. 
 
Propagating fronts connecting distorted hexagons to stripes or other distorted hexagons would be another interesting avenue to explore; see for instance~\cite{uecker2014,wetzel2018} for various stationary fronts between hexagons and stripes. The pattern selection mechanism in these fronts has yet to be explored even in the stationary case. 

From an algorithmic view, developing the far-field core decomposition idea to be used with time-stepping bifurcation routines would allow for the investigation of various invasion front instabilities found in doubly diffusive convection and plane Couette flow~\cite{beaume2018,pershin2018}. We note that the almost planar depinning fronts are similar to the twist instability seen by~\cite{beaume2018} close to the snaking region and where stability is regained as one moves sufficiently far from the snaking region. It would be interesting to see if one can detect this instability numerically. 

The invasion process of the stripe patches appears to be well explained by the parallel and perpendicular front computations. It would be interesting to see if there is a way to relate the evolution of the stripe patch to the Eikonal equation possibly using the ideas in~\cite{peletier2012a}. 
Perhaps an easier problem would be to understand the invasion process of radial ring patterns outside the ``snaking" region described in~\cite{McCalla2010,bramburger2018}.

In terms of other front propagation instabilities, the numerical methods outlined here could allow one to investigate the transition from pushed to pulled fronts (see~\cite{saarloos2003}) or modulated fronts emerging from a Turing bifurcation of a large front~\cite{sandstede2001}. Of course, the investigation of modulated fronts in  nonlocal equations~\cite{faye2015} and lattice systems~\cite{hupkes2009,bramburger2019} remain of major interest. 

\section*{Acknowledgements}
DJBL would like to thank Cedric Beaume, Jonathan Dawes, Edgar Knobloch, Bjorn Sandstede and Arnd Scheel for helpful discussions on this work. The author would like to thank the referees for their helpful and constructive comments that improved the paper. 

\appendix

\section{Appendix}

\subsection{Appendix: Proof of Proposition 1.1}\label{s:mod_fronts_inf}

In this section, we discuss why we expect generically that parallel stripe invasion fronts select a wavespeed and wavenumber whereas retreating parallel stripe fronts only select a wavespeed. 
To simplify the discussion and since parallel stripe fronts have no $y$-dependence, we consider the 1D SH equation. The results in this section are adaptions of those in Sandstede \& Scheel~\cite{sandstede2004,sandstede2008} and Goh \& Scheel~\cite{goh2017}.

The stationary parallel stripes of a depinning front in a co-moving frame are of the form $u_s(k_xx;k_x) = u_s(k_x(\xi\pm ct);k_x)$ and are time-periodic where $c>0$ is the propagation speed of the front and the $+$ denotes an invasion pattern forming front whereas $-$ denotes a retreating front. One-dimensional invasion fronts are of the form $u(x,t)=u(\rho,\tau)$, $\rho=k_x\xi,\tau=\omega t,\omega=ck_x$, solving the boundary value problem 
\begin{subequations}\label{e:1D_BVP_d}
\begin{align}
\omega(\pm u_{\rho} - u_{\tau}) - (1+k_x^2\partial_{\rho}^2)^2u - \mu u + f(u) =& 0,\\
u(\rho,\tau) - u(\rho,\tau+2\pi) =& 0,\\
\lim_{\rho\rightarrow\infty}u(\rho,\tau) =& 0,\\
\lim_{\rho\rightarrow-\infty}(u(\rho,\tau) - u_s(\rho\pm\tau;k_x)) =& 0,
\end{align}
\end{subequations}
for some $k_x>0$ and $\omega>0$ which we call the {\it wavenumber} and {\it transition frequency}, respectively. 

The linearisation of (\ref{e:1D_BVP_d}) about a front solution, $u^*$, is given by 
\[
\mathcal{L}u := \omega(\pm u_{\rho} - u_{\tau}) - (1+k_x^2\partial_{\rho}^2)^2u - \mu u + f_u(u^*)u.
\]
We note that $u^*$ solves a pseudo-elliptic equation, such that $\partial_\tau u$ and $\partial_\rho^4u$ belong to $L^\infty$ and $u^*$ can be readily seen to be smooth~\cite{goh2017}. 

As pointed out in~\cite{sandstede2004,goh2017}, the operator $\mathcal{L}$ is not Fredholm as a closed and densely defined operator on $L^2(\mathbb{R}\times\mathbb{T})$, say, where $\mathbb{T}=\mathbb{R}/2\pi\mathbb{Z}$ since both $\partial_\rho u^*$ and $\partial_\tau u^*$ lie in the kernel of $\mathcal{L}$ but they do not converge to zero at $\rho=\pm\infty$, such that a simple Weyl sequence construction shows that the range is not closed.

The operator $\mathcal{L}$ can be shown to be Fredholm by considering $\mathcal{L}$ as a closed operator on exponentially weighted spaces
\[
L_{\alpha}^2(\mathbb{R}\times\mathbb{T}) = \left\{u(\rho,\tau)\in L^2_{\mbox{loc}}(\mathbb{R}\times\mathbb{T})\; | \; e^{\alpha|\rho|}u\in L^2(\mathbb{R}\times\mathbb{T}) \right\},
\]
with small weights $\alpha\sim 0$. 

We define the asymptotic linear operators $\mathcal{L}_{l/r}$ where for $\mathcal{L}_{l}$ the $u^*$ is replaced by $u(\rho,\tau) = u_s(\rho\pm\tau;k_x)$ and for $\mathcal{L}_r$ the $u^*$ is replaced by zero. It is known that the operator $\mathcal{L}$ is Fredholm on $L_{\alpha}^2$ if and only if $\mathcal{L}_{l}$ is invertible on $L_{-\alpha}^2$ and $\mathcal{L}_{r}$ is invertible on $L_{\alpha}^2$; see for instance~\cite{sandstede2004}.  

We start by looking at the operator $\mathcal{L}_l$ and its dispersion relation. Substituting in $u_s(\rho\pm\tau)$ into (\ref{e:1D_BVP_d}), we see that $u_s(x)$ must be a $2\pi$-periodic solution of the ODE
\[
-(1+k_x^2\partial_x^2)^2u_s - \mu u_s + f(u_s) =0.
\]
Linearising this equation about $u_s$, we obtain the linear operator $\mathcal{L}_s$
\[
\mathcal{L}_s:=-(1+k_x^2\partial_x^2)^2 - \mu  + f'(u_s),
\]
which defines a closed operator on $L^2(0,2\pi)$ with domain $H^2_{\mbox{per}}(0,2\pi)$. We introduce the Floquet operator 
\[
\mathcal{\hat L} = -(1+[k_x\partial_x+\eta]^2)^2 - \mu  + f'(u_s(x)),
\]
on $L^2(\mathbb{R},\mathbb{C})$ with domain $H^4(\mathbb{R},\mathbb{C})$. 

\begin{hypo}\label{h:h1}
We assume there is an open region of $k_x$ where the periodic orbits are stable. Furthermore, the spectrum of $\mathcal{\hat L}$ on $L^2(\mathbb{R},\mathbb{C})$ lies in the open left half-pane except for a simple eigenvalue, $\hat\lambda$ close to the origin for $\eta\sim0$ with expansion
\begin{equation}\label{e:lam_lin}
\hat\lambda = d_{\|}\eta^2 + \mathcal{O}(\eta^4),
\end{equation}
with $d_{\|}>0$. 
\end{hypo}
For $\mu\sim0$, the results in \S\ref{s:roll_stab} suggest this is true for the stripes in the cubic-quintic SH equation by setting $(\sigma,\tau)=(i\eta,0)$ in proposition \ref{l:stab_rolls}. We anticipate this hypothesis is true in general for both the quadratic-cubic SH equation and for larger values of $\mu$.

Now a non-zero number $\rho\in\mathbb{C}$ is in the spectrum of $\mathcal{L}_l$ if, and only if, the linearised eigenvalue problem $\mathcal{L}_lu = \lambda u$ has a bounded non-zero solution $u(\rho,\tau)$ that is $2\pi$-periodic in $\tau$, where the Floquet multipler $\rho$ and the Floquet exponent $\lambda$ are related via $\rho = \exp(2\pi\lambda/\omega)$. These solutions can be calculated using the Floquet ansatz
\[
u(\rho,\tau) = e^{\eta\rho}w(\rho\pm\tau),
\]
where $w(\cdot)$ is $2\pi$-periodic and $\eta\in i\mathbb{R}$. Substituting this ansatz into $\mathcal{L}_l$ yields
\begin{equation}\label{e:lam2}
[\lambda \mp\omega \eta]w = \mathcal{\hat L}w.
\end{equation}
Comparing~(\ref{e:lam_lin}) and (\ref{e:lam2}), we see that $\hat\lambda$ is in the spectrum of the periodic orbit, computed in the travelling frame, if and only if
\[
\lambda = \hat\lambda \pm \omega \eta,
\]
is a Floquet exponent of $\mathcal{L}_l$, where $\eta\in i\mathbb{R}$ is the associated spatial Floquet exponent. In particular, for $\eta\in i\mathbb{R}$ close to zero, 
\begin{equation}\label{e:lam3}
\lambda =\hat\lambda(\eta) \pm \omega\eta = \pm \omega\eta + d_{\|}\eta^2 + \mathcal{O}(\eta^4).
\end{equation}

We next look at the spectrum of the operator $\mathcal{L}_r$. To do this, we linearise about the trivial state and include the homotopy parameter $\lambda$ to yield
\[
 \omega(\pm u_{\rho}-u_{\tau}) - (1+k_x^2\partial_\rho^2)^2u - \mu u = \lambda u.
\]
Using the ansatz $u=e^{i\ell \tau}e^{\eta\rho}\hat u$, yields the dispersion relation
\[
\lambda = \omega(\pm\eta - i\ell) - (1+k_x^2\eta^2)^2 - \mu.
\]

\begin{hypo}\label{h:convec}
We assume that the trivial state is linearly stable i.e., $\mu>0$.
\end{hypo}

\begin{remark}
From the numerics, hypothesis~\ref{h:convec} can be weakened to just requiring the trivial state is convectively unstable i.e., solutions of the linearized equation with spatially localized initial conditions decay pointwise, precisely when $\lambda$ lies in the left half-plane. In this paper, we are just concentrating on the bistable region where $\mu>0$. We note though that for $\mu\sim0$, we can explicitly compute the convective instability threshold. Setting $\eta_x=i\kappa/k_x,\kappa\in\mathbb{R}$, we have the dispersion relation 
\[
\tilde \lambda = i\kappa c - \mu - (1-\kappa^2)^2,
\]
where $\tilde\lambda = \lambda + i\ell\omega$. 
Following~\cite{doelman2003}, $\tilde\lambda$ has a double root at $(\kappa,\mu,c)=(\pm1,0,0)$. This double root persists for $(\mu,c)\neq0$ and is, to leading order given by $\kappa_*=\pm1 +ic/8$. The real part of $\tilde\lambda$ evaluated at $\kappa_*$ is given by to leading order $\Re\tilde\lambda(\kappa_*)=\Re\lambda(\kappa_*)=-\mu-c^2/16$. Hence for $\mu>0$ (i.e. in the bistable region) there is no restriction on $c$ whereas for $\mu<0$, we require $c>4\sqrt{-\mu}$. 

This additional assumption allows for the existence of invading fronts that select a unique propagation speed $c$ and far-field wavenumber $k_x$ for fixed $\mu<0$. Numerically, we can continue invasion fronts into the region $\mu>0$.
\end{remark}

\begin{lemma}\label{L:fred} (Fredholm Crossing). Assume that $\mu>0$ and $\omega\neq 0$, then the operator $\mathcal{L}$ is Fredholm with index $-1$ for invading fronts and $0$ for retreating fronts 
in $L_{\alpha}^2(\mathbb{R}\times\mathbb{T})$, $\alpha>0$, sufficiently small.
\end{lemma}
\begin{proof}
We follow the proof~\cite[Lemma 2.1]{goh2017} and~\cite[Lemma 3.6]{sandstede2004} and characterise the Fredholm indices using Fredholm borders. The Fredholm index can be computed by counting the signed crossings of multipliers through the origin during a homotopy from $\mathcal{L}_r(\eta)$ to $\mathcal{L}_l(-\eta)$; see~\cite[Lemma 2.1]{goh2017}. 
We note that for $\mu>0$, the asymptotic (trivial) state at $\rho=+\infty$ is linearly stable from hypothesis \ref{h:convec}. 
Provided $\omega\neq 0$, then~(\ref{e:lam3}) has a simple spatial Floquet exponent $\eta=\eta(\lambda)$ for all $\lambda$ close to zero and 
\[
\left.\frac{\D \eta}{\D \lambda}\right|_{\lambda=0} = \pm\frac{1}{\omega}.
\]
Hence, for $\lambda<0$ close to zero, the relative Morse index at $\rho=-\infty$ is therefore $+1$ for invading fronts and $-1$ if retreating fronts. Furthermore, for invading fronts there are no crossings if we homotope between $\mathcal{L}_r(\eta)$ and $\mathcal{L}_l(\eta)$ but for retreating fronts there must be another crossing with sign $+1$. Hence $\mathcal{L}$ is Fredholm of index 0 for $\alpha<0$, small, and Fredholm of index $-1$ for invading or $0$ for retreating fronts, respectively, for $\alpha>0$, small.
\end{proof}

\begin{hypo}\label{h:trans} (Transverse front)
Assume that the kernel of $\mathcal{L}$ in $L_{-\alpha}^2$, $\alpha>0$ sufficiently small is two-dimensional spanned by $\partial_\rho u^*$ and $\partial_\tau u^*$. Furthermore, the eigenvalue $\lambda=0$ is algebraically simple. 
\end{hypo}

We also define the $L^2$-adjoint of  $\mathcal{L}$,
\[
\mathcal{L}^{\mbox{ad}}:X_0\subset L^2 \rightarrow L^2,\qquad \mathcal{L}^{\mbox{ad}}v = -\omega(\pm v_\rho - v_\tau) - (1+k_x^2\partial_\rho^2)^2v - \mu v + f_u(u^*)v,
\]
where $X_0=H^1(\mathbb{T},L^2(\mathbb{R}))\cap L^2(\mathbb{T},H^4(\mathbb{R}))$. 
Hence, the co-kernel of $\mathcal{L}$ is also spanned by $\partial_\rho u^*$ and $\partial_\tau u^*$.  

It turns out that a linear combination of the kernel elements form a localised eigenfunction. 

\begin{lemma} 
There exists constants $\delta>0$ and $\psi\in\mathbb{R}$ such that 
\[
u^*(\rho) = u_s(\rho\pm\cdot + \psi;k_x) + \mathcal{O}(e^{\delta\rho}),
\]
in $H^k(\mathbb{R}\times\mathbb{T})$ as $\rho\rightarrow-\infty$. 
The same estimate is true for the derivatives with respect to $\rho$ and $\tau$. 
Furthermore, the geometric multiplicity of $\lambda=0$ as an eigenvalue of the point spectrum of $\mathcal{L}$ posed on $L^2_{\alpha}$ is 1.  
\end{lemma}
\begin{proof}
The exponential convergence and asymptotic phase of the invasion front to $u_s$ have been proved in \cite[Lemma 2.2]{goh2017} using the spatial dynamic arguments in~\cite[Theorem 3]{sandstede2004}. Using that 
\[
\partial_{\rho}u^*(\rho,\cdot) = u_s'(\rho\pm\cdot) + \mathcal{O}(e^{\delta\rho}), \qquad
\partial_{\tau}u^*(\rho,\cdot) = u_s'(\rho\pm\cdot) + \mathcal{O}(e^{\delta\rho}),
\]
for some $\delta>0$, as $\rho\rightarrow-\infty$, we can find an exponentially localised eigenfunction as $|\rho|\rightarrow\infty$. This eigenfunction is given by $(\pm\partial_{\rho}-\partial_{\tau})u^*(\rho,\cdot)$ and generates a one-dimensional subspace of solutions that decay exponentially with rate $\delta$.
\end{proof}

Since one of the eigenfunctions is exponentially localised as $\rho\rightarrow-\infty$ 
(which we call $e_{\omega}$ with corresponding adjoint eigenfunction which is also $e_{\omega}$) 
and the other is not, then the kernel of $\mathcal{L}$ is one-dimensional in $L_{\alpha}^2,\alpha\gtrsim0$. 

To show persistence of an invasion front, we employ the far-field core decomposition~\cite{lloyd2017,goh2017}
\begin{equation}\label{e:far_core_1D}
u(\rho,\tau) = u_s(\rho\pm\tau + \psi;k_x)\chi(\rho) + v(\rho,\tau; \omega),
\end{equation}
where $v\in X_\alpha :=H^1(\mathbb{T},L_\alpha^2(\mathbb{R}))\cap L^2(\mathbb{T},H_\alpha^4(\mathbb{R}))$, and $k_x,\psi,\omega$ are free variables. 
Substituting this ansatz into the SH equation yields (after subtracting the equation for $u_s$)
\begin{equation}\label{e:PDE_w}
\mathbb{L}[u_s(\rho\pm\tau+\psi;k_x)\chi(z) + v(\rho,\tau;\omega)] + f(u_s(\rho\pm\tau+\psi;k_x)\chi(z) + v(\rho,\tau;\omega)) - \chi\left(\mathbb{L}u_s + f(u_s) \right) = 0,
\end{equation}
where $\mathbb{L}u :=  \omega(\pm u_{\rho} - u_{\tau}) - (1+k_x^2\partial_{\rho}^2)^2u - \mu u$. 

We consider the left-hand side of (\ref{e:PDE_w}) as a (locally defined) nonlinear operator
\begin{equation}\label{e:Fv}
F_v \; : X_{\alpha}\times \mathbb{R}^3\rightarrow L^2_{\alpha},\qquad (v,k_x,\psi,\omega)\rightarrow F_v(v,k_x,\omega).
\end{equation}
Note that $F_v$ is well defined since terms not involving $v$ are given by commutators between cut-off and differential operators and nonlinearities, hence compactly supported. Moreover, $F_v$ is readily seen to be a smooth function and the the derivative with respect to $v$ at a front $u^* = v^*+ \chi u_s$ is the linearisation 
\[
\partial_vF_v = \mathcal{L},
\]
so that $DF_v(v,k_x,\psi,\omega) = (\partial_vF_v(v^*)v,\partial_{k_x}F_v,\partial_\psi F_v, \partial_\omega F_v)$  is Fredholm of index 2 for invading fronts and index 3 for retreating fronts by Fredholm bordering theory. 
We now prove Proposition~\ref{prop:1}. 

\begin{Proposition}
Assume Hypotheses \ref{h:h1}-\ref{h:trans}, and there exists a front solution $(v^*,k_x^*,\psi^*,\omega^*)$ such that $F_v(v^*;k_x^*,\psi^*,\omega^*)=0$, then invading fronts with $\omega<0$ select a unique far-field wavenumber $k_x$ and transition frequency $\omega$ and retreating fronts with $\omega>0$ select just a transition frequency $\omega=\omega(k_x)$.
\end{Proposition}
\begin{proof} 
We prove persistence using Lyapunov-Schmidt reduction following~\cite[Theorem I.2.3 and I.4.1]{kielhofer2012}. 

We assume the decompositions 
\[
X_\alpha = \mbox{ker}(DF_v)\oplus X^0,\qquad L^2_\alpha = \mbox{ran}(DF_v)\oplus L^0,
\]
where $X^0= \mbox{ker}(DF_v)^\perp$ and $L^0= \mbox{ran}(DF_v)^\perp$ and we define the orthogonal projection $Pu = \langle e_{\omega},u \rangle_{L_\alpha^2} e_{\omega}$ (the decomposition is valid since solutions of (\ref{e:1D_BVP_d}) are regular). We decompose  $v$ given by
\[
v(\rho,\tau;\omega)= \beta e_{\omega} + \tilde v(\rho,\tau;\omega),\qquad \beta\in\mathbb{R},\quad \tilde v\in X_\alpha
\]
where $P\tilde v =0$  and we re-write the nonlinear problem as
\begin{align*}
PF_v\left(\beta e_{\omega} + \tilde v(\rho,\tau;\omega),k_x,\omega\right) =0, \qquad
(I-P)F_v\left(\beta e_{\omega} + \tilde v(\rho,\tau;\omega),k_x,\omega\right) =0.
\end{align*}
Via a simple adaption of \cite[Lemma 6.5]{lloyd2017}, it can be shown that the linear operator $(I-P)DF_v$ has a trivial kernel and is invertible. Hence the second equation can be solved using the implicit function theorem for $(\tilde v,k_x) = G(\omega,\psi,\beta)$ for invading fronts or $\tilde v = G(\omega,k_x,\psi,\beta)$ for retreating fronts. We now consider the first equation $\tilde F_v$ given by
 \[
 \tilde F_v := PF_v\left(\beta e_{\omega} +\tilde v,k_x,\omega\right) =0,\qquad \tilde F_v :  \tilde X_\alpha\times I_{\psi^*,\beta^*}\times W_{\omega^*} \rightarrow \mbox{ran}(L^0)
 \]
where we take neighbourhoods $\tilde X_\alpha$ of zero in the kernel of $F_v$, $W_{\omega^*}$ for values of $\omega$ near $\omega^*$, and $I_{\psi^*,\beta^*}$ for values $(\psi,\beta)$ near $(\psi^*,0)$.
Since $\partial_\omega \tilde F_v\notin R(\mathcal{L})$ (from the assumption of the zero eigenvalue being algebraically simple), we can apply the implicit function theorem to solve for $\omega=H(\tilde v,k_x,\psi,\beta)$. Since $\beta$ corresponds to spatial translations of the front and is arbitrary, we can set $\beta=0$. Hence, we find for invading fronts a unique continuation of $k_x$ and $\omega$ whereas for retreating fronts $\omega = \omega(k_x)$. 
\end{proof}

\subsection{Approximation of parallel invasion fronts on finite domains}\label{s:modulated_finite}
We now look at approximating parallel invasion fronts on finite domains in $\rho$ since the previous section suggests that for invading fronts there is an isolated invasion front on the infinite line. We compute an invasion fronts on $\Omega_{L_\rho}=(-L_\rho,L_\rho)\times(0,2\pi)$ and we approximate $F_v$~(\ref{e:Fv}) on $v\in X(\Omega_{L_\rho}):=H^1((0,2\pi),L^2(-L_\rho,L_\rho))\cap L^2((0,2\pi),H^4(-L_\rho,L_\rho))$ with periodic boundary conditions in $\tau$ and and boundary conditions at $\rho=\pm L_{\rho}$ to be specified below. We denote this approximation by $F_v^{L_{\rho}}$. We note that while the linearised operator $\partial_vF_v^{L_{\rho}}$ with appropriate boundary conditions is Fredholm with index zero, it is however very ill-conditioned; see~\cite[\S6.4]{lloyd2017}. 

To overcome this, we impose the boundary conditions on $v$ at $\rho=\pm L_\rho$, for all $\tau\in[0,2\pi)$, such that $-\Delta^2$ is Fredholm of index 0 with the same boundary conditions e.g., we impose Dirichlet boundary conditions in $v$. We also introduce the following phase conditions
\begin{equation}\label{e:phase_cons}
\Phi_1 v =\int_0^{2\pi}\int_{-L_\rho}^{L_\rho}([\partial_\rho-\partial_\tau]v^{\mbox{old}})(v-v^{\mbox{old}})\D\rho\D\tau,\qquad \Phi_2v = \int_0^{2\pi}\int_{-L_\rho}^{-L_\rho+2\pi/|k_x|}u'_sv\D\rho\D\tau,
\end{equation}
where $v^{\mbox{old}}$ is a template solution (e.g. the initial guess or previous solution). 
The first phase condition is the standard phase condition for travelling fronts for selecting the invasion speed and has been proven in~\cite[Chapter 11]{krauskopf2007} to be an admissible phase condition due to the function $(\partial_\rho + \partial_\tau)u^{\mbox{old}}$ decaying exponentially fast to zero. The second phase condition is the same one used to select the far-field wavenumber $k_x$ for stationary grain boundaries; see~\cite{lloyd2017}. 
We then solve the system $(F_v^{L_\rho},\Phi_1v,\Phi_2v)=\mathbf{0}$ for $(v,k_x,\omega)$. The linear operator $(\partial_{v,k_x,\omega}F^{L_\rho}_v,\Phi_1, \Phi_2)\::\: X_{\mbox{bc}}\times \mathbb{R}^2\rightarrow L^2\times\mathbb{R}^2$ is Fredholm with index $0$.  Under the assumption of transverse boundary conditions~\cite[Hypothesis 6.8]{lloyd2017} it should be possible to show exponential convergence in $L_\rho$ to the infinite-dimensional system.

\section{Appendix}\label{s:lay_proof}

\subsection{Proof of Proposition 2.1: Stripe existence}
We obtain the existence of stripes for the cubic-quintic SH equation for $(\mu,\nu)\approx(0,0)$ via Lyapunov-Schmidt reduction following~Mielke~\cite{mielke1995a}. Omitting tilde's and redefining $x$, we define nonlinear problem $F(\mu,\kappa,\nu,u)=0$, where
\begin{equation}\label{e:non_F_appendix}
F(\mu,\kappa,\nu,u) := -(1+(1+\kappa)\partial_x^2)u - \mu u + \nu u^3 - u^5,\qquad F:\mathbb{R}^4\times X^4\rightarrow X^0,
\end{equation}
$X^j = H^j_{\mbox{per}}([0,2\pi],\mathbb{R})$ and $k_x = \sqrt{1+\kappa}$. 

The linearisation about the $u=0$ is given by $L_{\mbox{per}}u = \partial_uF(0,0,0,0)[u]=-(1+\partial_x^2)u$ is self-adjoint with discrete spectrum $\Sigma(L_{\mbox{per}})=\{-(1-n^2)^2 : n\in\mathbb{Z}\}$ and kernel $\mbox{ker}(L_{\mbox{per}})$ spanned by $U_1=\cos(x)$ and $U_2=\sin(x)$. We assume the decompositions
\[
X^4 = \mbox{ker}(L_{\mbox{per}}) \oplus \mbox{ker}(L_{\mbox{per}})^\perp,\qquad \mbox{and}\qquad X^0 = \mbox{ran}(L_{\mbox{per}}) \oplus  \mbox{ran}(L_{\mbox{per}})^\perp,
\]
where $\mbox{ran}(L_{\mbox{per}})$ is the range of $L_{\mbox{per}}$. We define the orthogonal projection $Pu = \frac{1}{2\pi}(\langle U_1,u\rangle U_1 + \langle U_2,u\rangle U_2)$ where $\langle u,v\rangle = \int_0^{2\pi}uv\D x$ and the mapping $Q:\: X^j\rightarrow\mathbb{R}^2$ where $Qu = \frac{1}{2\pi}(\langle u,U_1\rangle,\langle u,U_2\rangle)^T$. 

We decompose $u\in X^4$ into $\alpha_1 U_1 + \alpha_2 U_2 + V$, where $PV = 0$. We re-write nonlinear problem~(\ref{e:non_F_appendix}) as
\begin{align}
QF(\mu,\kappa,\tilde\nu,\alpha_1 U_1 + \alpha_2 U_2 + V) =& \mathbf{0},\\
(I-P)F(\mu,\kappa,\tilde\nu,\alpha_1 U_1 + \alpha_2 U_2 + V) =& 0.
\end{align}
The second equation can be solved for $V\in (I-P)X^4$ as a function of $(\mu,\kappa,\nu,\alpha)$ using the implicit function theorem. From $(I-P)L_{\mbox{per}}U_j=0$, we find that $V=G(\mu,\kappa,\nu,\alpha) = \mathcal{O}(|\nu||\alpha|^3+|\alpha|^5)$ for $\alpha\rightarrow0$. Hence, all small periodic solutions of (\ref{e:non_F}) satisfy the bifurcation equation
\[
f(\mu,\nu,\kappa,\alpha) :=QF(\mu,\nu,\kappa,\alpha_1 U_1 + \alpha_2 U_2 + V) = \mathbf{0}\in\mathbb{R}^2,
\]
where $f$ is given by
\[
f(\mu,\nu,\kappa,\alpha) = \left[-\frac12\kappa^2 - \frac12\mu + \left(\frac{3\nu}{8}  + \mathcal{O}(|\mu||\nu|)\right)|\alpha|^2 - \left(\frac{5}{16}+\mathcal{O}(|\mu| + |\nu|)\right)|\alpha|^4 + \mathcal{O}(|\alpha|^6)\right]\alpha,
\]
where $|\alpha|^2 = \alpha_1^2 + \alpha_2^2$, for $\alpha_1,\alpha_2\rightarrow0$ uniformly for $(\mu,\kappa,\nu)$ small. Solving the pre-factor in front of $\alpha$ in $f$ leads to the existence result. 

\subsection{Proof of Proposition 2.3: Stripe stability}

Linear stability of the periodic solutions to (\ref{e:non_F_appendix}) can be found by looking at the linear problem
\[
-(1+\partial_x^2 + \partial_y^2)^2w + \mu w + (3\tilde\nu\epsilon^2\tilde u_{\mu,k,\tilde\nu}^2(x) - 5\tilde u_{\mu,\kappa,\tilde\nu}^4(x))w = \lambda w.
\]
Carrying out a Floquet-Bloch decomposition of the form $w = e^{i(\sigma x + \tau y)}\tilde w(\xi)$ where $\xi = k_xx$ and $\tilde w\in X^4$, leads to the eigenvalue problem
\begin{equation}\label{e:roll_lin_op_appendix}
L(\mu,k_x,\nu,\sigma,\tau,\lambda)\tilde w := -(1+(k_x\partial_\xi + i\sigma)^2 - \tau^2)^2\tilde w + \mu \tilde w+ (3\tilde\nu\epsilon^2\tilde u_{\mu,\kappa,\nu}^2(x) - 5\tilde u_{\mu,\kappa,\nu}^4(x))\tilde w  - \lambda \tilde w = 0
\end{equation}
 We solve this eigenvalue problem using Lyapunov-Schmidt reduction and decompose $\tilde w = \beta_1 U_1 + \beta_2 U_2 + W$ with $PW=0$. For sufficiently small $(\mu,k_x-1,\nu,\sigma,\tau,\lambda)$, the equation $(I-P)L(\mu,k_x,\nu)\tilde w = 0$ can be solved uniquely for $W=\mathcal{W}(\epsilon,k_x,\nu)\beta=\mathcal{O}((\nu a^2 + a^4)|\beta|)$ as $a\rightarrow0$. In particular,
 $\mathcal{W}$ satisfies 
 \[
 \mathcal{W}(\mu,k_x,\nu,\sigma,\tau,\lambda) = D_{\alpha}G(\mu,k_x,\nu,(\tilde a,0)) + i\sigma H(\mu,k_x,\nu,\tilde a) + \tilde a^2\mathcal{O}(\sigma^2+\tau^2+|\lambda|),
 \]
 where $D_{\alpha}G(\mu,k_x,\nu,(\tilde a,0)) = \mathcal{O}(|\nu|\tilde a^2 + \tilde a^4)$ and $H(\mu,k_x,\nu,\tilde a) =\mathcal{O}(|\nu| \tilde a^2 + \tilde a^4)$ since the coupling of $\beta$ and $W$ in the equation for $W$ is only active in the $(3\nu\tilde u_{\mu,\kappa,\nu}^2 - 5\tilde u_{\mu,\kappa,\nu}^4(x))w$ term. Inserting $\mathcal{W}$ into $QL(\epsilon,k_x,\tilde\nu,\sigma,\tau,\lambda)[\beta_1U_1 + \beta_2U_2 + \mathcal{W}\tilde w ]= \mathbf{m}(\epsilon,k_x,\nu,\sigma,\tau,\lambda) $ yields the eigenvalue problem $\mathbf{m}=0$, where
 \[
 \mathbf{m}(\mu,k_x,\nu,\sigma,\tau,\lambda) =  \begin{pmatrix}
 \rho + c(\mu,\nu,k_x) -\lambda & i\delta \\
 -i\delta & \rho-\lambda
 \end{pmatrix} + \tilde a^4\begin{pmatrix}
 \mathcal{O}(\sigma^2+\tau^2+|\lambda|) &  \mathcal{O}(|\sigma|+\tau^2+|\lambda|) \\
 \mathcal{O}(|\sigma|+\tau^2+|\lambda|) &  \mathcal{O}(\sigma^2+\tau^2+|\lambda|)
 \end{pmatrix},
 \]
 and 
 $\rho = -2k_x^2\sigma^2+(1-k_x^2)(\sigma^2+\tau^2)-\frac12(\sigma^2+\tau^2)^2,\delta=2k\sigma(\sigma^2+\tau^2+k^2-1)$ and $c(\mu,\nu,k_x) = \frac34\nu\tilde a^2 - \frac54\tilde a^4+\mathcal{O}(\tilde a^6)$. 
 
  In order to determine the most unstable perturbations, we introduce the scalings
 \[
\mu = \epsilon^4\hat\mu,\qquad (k_x^2-1)=2\epsilon^2\hat\kappa,\qquad\sigma = \epsilon^2\hat\sigma,\qquad\tau = \epsilon\hat\tau,\qquad\lambda = \epsilon^4\hat\lambda,
 \]
 where $\hat\kappa \in[-\frac{\tilde\kappa}{2},\frac{\tilde\kappa}{2}]$ and $|\epsilon|\ll 1$. 
 
 We then find 
 \[
 \mathbf{m}(\epsilon,\hat\mu,\sqrt{1+2\epsilon^2\hat\kappa},\epsilon^2\hat\nu,\epsilon^2\hat\sigma,\epsilon\hat\tau,\epsilon^4\hat\lambda) = \epsilon^4\mathbf{\widehat m}(\tilde\mu,\hat\kappa,\hat\nu,\hat\sigma,\hat\tau,\hat\lambda)+ \mathcal{O}(\epsilon^5),
 \]
 where 
 \[
 \mathbf{\widehat m}(\hat\mu,\hat\kappa,\hat\nu,\hat\sigma,\hat\tau,\hat\lambda) = \begin{pmatrix}
 -2\hat\sigma^2-2\hat\kappa\hat\tau - \frac12\hat\tau^4 + \frac34\hat\nu\hat a^2_\pm - \frac54\hat a^4_\pm - \hat\lambda & 2i\hat\sigma(\hat\tau^2+2\hat\kappa) \\
 -2i\hat\sigma(\hat\tau^2+2\hat\kappa) & -2i\hat\sigma^2-2\hat\kappa\hat\tau - \frac12\hat\tau^4 -\hat\lambda
 \end{pmatrix},
 \]
 and 
 \[
 \hat a^2_\pm = \frac{3\hat\nu}{5}\pm\frac{\sqrt{9\hat\nu^2-40(\tilde\mu+4\hat\kappa^2)}}{5} + \mathcal{O}(\epsilon^4).
 \]
 The eigenvalues of $\mathbf{\hat m}$ are given by
 \begin{align}
 \hat\lambda_\pm(\hat\mu,\hat\kappa,\hat\nu,\hat\sigma,\hat\tau) =& \frac12\left[2(-2\hat\sigma^2-2\hat\kappa\hat\tau - \frac12\hat\tau^4 )+ \frac34\hat\nu\hat a^2_\pm - \frac54\hat a^4_\pm\right]\\&\pm\frac12\sqrt{16(\hat\tau^2+2\hat\kappa)\hat\sigma^2+(\frac34\hat\nu\hat a^2_\pm - \frac54\hat a^4_\pm)^2}.
 \end{align}
The critical eigenvalue for the upper branch $\hat a_+$ is maximal with respect to $(\hat\sigma,\hat\tau)$ on the axes $\hat\sigma=0$ or $\hat\tau=0$ and is given by
\begin{align}
\lambda_+(\hat\mu,\hat\kappa,\hat\nu,0,\hat\tau) =& (-2\hat\kappa\hat\tau - \frac12\hat\tau^4 )\leq\left\{ \begin{array}{cl}
0 & \mbox{for $\kappa\in[0,\frac{\tilde\kappa}{2}]$,}\\
4\hat\kappa^2 & \mbox{for $\hat\kappa\in[-\frac{\tilde\kappa}{2},0]$,}
\end{array}\right.
\\
\lambda_+(\hat\mu,\hat\kappa,\hat\nu,\hat\sigma,0) =& \frac12\left[2(-2\hat\sigma^2)+ \frac34\hat\nu\hat a^2_\pm - \frac54\hat a^4_\pm\right]+\frac12\sqrt{16(2\kappa)\hat\sigma^2+(\frac34\hat\nu\hat a^2_\pm - \frac54\hat a^4_\pm)^2},\\
&\leq\left\{ 
\begin{array}{cl}
0 & \mbox{for $\kappa^2\leq\frac{135\hat\nu^2-800\tilde\mu+15\hat\nu\sqrt{81\hat\nu^2-320\tilde\mu}}{6400}$}, \\
\frac{(8\kappa^2+c)^2}{32\kappa^2} & \mbox{for $\kappa^2\in\left[\frac{135\hat\nu^2-800\tilde\mu+15\hat\nu\sqrt{81\hat\nu^2-320\tilde\mu}}{6400},\frac{\tilde\kappa}{4}\right]$}.
\end{array}
\right.
\end{align}

\bibliographystyle{siam}
\bibliography{swift_depinning_bib}

\end{document}